\documentclass[pra,showpacs,twocolumn]{revtex4-1}

\usepackage{amsthm}
\usepackage{amsmath}
\usepackage{amssymb}
\usepackage[pdftex]{graphicx}
\usepackage{dsfont}
\usepackage{color}
\usepackage{fullpage}

\newtheorem{theorem}{Theorem}
\newtheorem{proposition}{Proposition}
\newtheorem{lemma}{Lemma}

\newcommand{\indep}{\rotatebox[origin=c]{90}{$\models$}}

\newcommand{\pr}[1]{\textrm{\textnormal{Pr}}\left(#1\right)}

\newcommand{\expt}{\mathbb{E}}

\newcommand{\bellmax}{\beta_{\rm max}}
\newcommand{\bellmin}{\beta_{\rm min}}
\newcommand{\bellwin}{\beta_{\rm win}}

\def\p{\textrm{Pr}}

\newcommand{\h}{h}
\newcommand{\HH}{H}
\newcommand{\HHTot}{H}

\newcommand{\NNN}{{m}}
\newcommand{\N}{m}
\newcommand{\ellInd}{i}

\newcommand{\mean}{\mathbb{E}}
\newcommand{\longS}{\mathcal{T}_{j\rightarrow i}}
\newcommand{\shortSindex}[1]{\mathcal{T}_{{#1}\rightarrow i}^i}
\newcommand{\shortS}{\mathcal{T}_{j\rightarrow i}^i}
\newcommand{\mtotalHat}{\mathbf{\hat{t}}^{\N}}

\newcommand{\bellF}{C}
\newcommand{\mbigCi}{\mathbf{C}^{i-1}}
\newcommand{\mvi}{\mathbf{c}^{i-1}}
\newcommand{\mbigTi}{\mathbf{T}^i}
\newcommand{\mbigTiprevious}{\mathbf{T}^{i-1}}

\newcommand{\mbigTtotal}{\mathbf{T}^{\N}}
\newcommand{\mti}{\mathbf{t}^i}
\newcommand{\mtiprevious}{\mathbf{t}^{i-1}}
\newcommand{\mtotal}{\mathbf{t}^{\N}}

\newcommand{\B}{\mathbb{B}_{\gamma}}
\newcommand{\Natural}{\mathbb{N}}

\newcommand{\deltaDef}{\delta}
\newcommand{\LHVMs}{\textrm{LHVM} }
\newcommand{\LHVM}{\textrm{LHVM}}

\newcommand{\pvalue}{$P$-value}
\newcommand{\pvalues}{$P$-value }
\newcommand{\myprob}{q}
\newcommand{\win}{\textrm{win}}
\newcommand{\lose}{\textrm{lose}}
\newcommand{\TT}{attempt}
\newcommand{\TTs}{attempts}
\newcommand{\ev}{trial}
\newcommand{\evs}{trials}

\newcommand{\heraldingEvent}{event-ready signal}

\begin{document}
\title{(Nearly) optimal $P$-values for all Bell inequalities}
\author{David Elkouss}
\address{QuTech, Delft University of Technology, Lorentzweg 1, 2628 CJ Delft, The Netherlands}
\author{Stephanie Wehner}
\address{QuTech, Delft University of Technology, Lorentzweg 1, 2628 CJ Delft, The Netherlands}
\email{s.d.c.wehner@tudelft.nl}

\begin{abstract}
A key objective in conducting a Bell test is to quantify the statistical evidence against a local-hidden variable model (LHVM) given that we can collect only a finite number of \evs\ in any experiment.
The notion of statistical evidence is thereby formulated in the framework of hypothesis testing, 
where the null hypothesis is that the experiment can be described by an LHVM. The statistical confidence with which the null hypothesis of an LHVM is
rejected is quantified by the so-called \pvalue, where a smaller \pvalues implies higher confidence. 
Establishing good statistical evidence is especially challenging if the number of \evs\ is small, or the Bell violation very low.
Here, we derive the optimal \pvalues for a large class of Bell inequalities. 
What's more, we obtain very sharp upper bounds on the \pvalues for \emph{all} Bell inequalities.
These values are easily computed from experimental data, and are valid even if we allow arbitrary memory in the devices. 
Our analysis is able to deal with imperfect random number generators, and event-ready schemes, even if such a scheme can create different
kinds of entangled states. Finally, we review requirements for sound data collection, and a method for combining $P$-values of independent experiments. The methods discussed here are not specific to Bell inequalities. For instance, they can also be applied to the study of certified randomness or to tests of noncontextuality.
\end{abstract}

\maketitle

%
\section{Introduction}

Local hidden variable models (\LHVM) predict concrete limitations on the statistics that can be observed in a Bell experiment~\cite{Bell_04}. These are typically phrased in terms of probabilities or expectation values. 
However, in any experiment we can only observe a finite number of trials, and not probabilities. We thus need to quantify the statistical evidence against an \LHVMs given a finite number of trials. 

The traditional way to analyze statistics in Bell experiments is to compute the number of standard deviations that separate the observed data from the best \LHVM. 
However, it is now known that this method has flaws~\cite{Gill_03b,Barrett_02,Zhang_11,Bierhorst_14} (see~\cite{Zhang_11} for a detailed discussion). In particular, 
we would have to assume Gaussian statistics and independence between subsequent \TTs, allowing for the memory loophole~\cite{Barrett_02,Gill_03b}. Fortunately, 
it is possible to rigorously analyze the statistical confidence even when allowing for memory as was first done by Gill~\cite{Gill_03a}. This is the approach that we follow here.

Instead of bounding the standard deviation, the intuitive idea behind the rigorous analysis is to bound the probability of observing the experimental data if nature was indeed governed by an \LHVM. 
In the language of hypothesis testing, this is known as the \pvalue, where the null hypothesis is that the experiment can be modelled as an \LHVMs (see e.g.~\cite{Vandam_05}).
Informally, we thus have
\begin{align}
&P\textrm{-value} = \max_{\LHVM}\nonumber\\
&\Pr[\mbox{data at least as extreme as observed}\nonumber\\
&\qquad \mid \mbox{experiment is governed by \LHVM}]\ .
\end{align}
A small \pvalues can be interpreted as strong evidence against the null hypothesis. Hence, in the case of a Bell experiment, a small \pvalues can be regarded as strong evidence against the hypothesis that the experiment was governed by an arbitrary \LHVM. 

There is an extensive literature regarding methods for evaluating the \pvalues in Bell experiments \cite{Peres_00,Barrett_02,Gill_03a,Gill_03b,Larsson_04,Acin_05, Vandam_05,Pironio_10,Zhang_10,Zhang_11, Pironio_13,Zhang_13,Bierhorst_14,Gill_14,Bierhorst_15} and discussions regarding the analysis of concrete experiments and loopholes \cite{Aspect_82,Weihs_98,Tittel_99,Rowe_01,Matsukevich_08,Ansmann_09,Scheidl_10,Giustina_13,Christensen_13,Pope_13,Bancal_14,Kofler_14,Larsson_14, Larsson_14b,Christensen_15,Knill_15}. 
Previous approaches to obtain such $P$-values known from the literature can be roughly divided into two categories. 
In the first approach, we select a suitable Bell inequality based on the expected experimental statistics
or test data collected ahead of time. After a Bell inequality is fixed, one can model the process as a (super-)martingale to which standard concentration 
inequalities~\cite{Gill_03a,Gill_03b,Larsson_04,Acin_05,Pironio_10, Pironio_13,Gill_14} can be applied. While this allows one to obtain bounds for all Bell inequalities relatively 
easily, the resulting upper bounds on the $P$-values are generally very loose. Crucially, 
this means that a much larger amount of trials would need to be collected than is actually necessary to obtain good statistical confidence. 
Figure~\ref{fig:chshcomp} 
illustrates the significance of using bounds employed in previous works compared to the bound used here.
When making a statement about all Bell inequalities below, we will also take a martingale
approach using however the much sharper concentration offered by the Bentkus' inequality~\cite{Bentkus_04}.
For some simple inequalities like Clauser-Horne-Shimony-Holt (CHSH) 
~\cite{Clauser_69} and Clauser-Horne (CH)~\cite{Clauser_74}, tight bounds on the \pvalues have been obtained 
when the measurement settings in the experiment are chosen uniformly, and no event-ready scheme is employed~\cite{Barrett_02,Bierhorst_14,Bierhorst_15}. Such a bound was first informally derived in \cite{Barrett_02}, and later rigorously developed by Bierhorst \cite{Bierhorst_14} whose approach for CHSH closely inspires our analysis of Bell inequalities that correspond to win/lose games below. 

\begin{figure}  
\begin{center}\includegraphics[width=8cm]{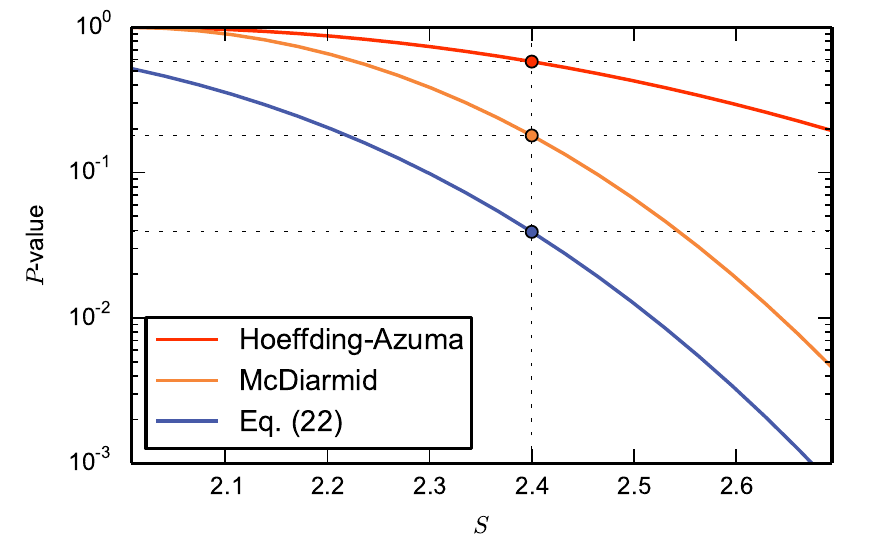}\end{center}
\caption{Comparison of $P$-value bounds for the CHSH inequality for values used in the first loophole free Bell test~\cite{Hensen_15}: The three curves show bounds on the \pvalue\ for a fixed number of \evs\ $n=245$ and random number generators bias $\tau=1.08\cdot 10^{-5}$~\cite{Abellan_14,Abellan_15}. The \pvalue\ is computed as a function of the violation $S$ which is defined as: $S=8(c/n-1/2)$, where $c$ is the number of wins in the CHSH game. 
From top to bottom, the curves show the bound on the \pvalues computed with Azuma-Hoeffding used in~\cite{Pironio_10}, McDiarmid's   
inequality~\cite{Mcdiarmid_89} given in~\cite{Zhang_13} and the upper bound from \eqref{eq:HeraldingBound} (with $\beta_\win=3/4+\tau-\tau^2$ as shown in Lemma 1 in the Appendix). 
In the Delft experiment \cite{Hensen_15} a number $c=196$ of wins were observed, giving $S=2.4$ and \eqref{eq:HeraldingBound} yields \pvalue$\approx 0.039$. The dots indicate the $P$-values
predicted by the other bounds. 
To obtain the same \pvalue\ with McDiarmid's inequality and Azuma-Hoeffding \cite{Pironio_10} the required violations would be $S=2.54$ and $S=2.98$ (beyond QM) respectively.}
\label{fig:chshcomp}
\end{figure}

The second approach that has been pursued is to combine the search for a good Bell inequality with a numerical method adapting to the data~\cite{Zhang_10,Zhang_11,Zhang_13}.
This method is asymptotically optimal in the limit of many experimental trials.
While conceptually beautiful, this numerical method can need a rather significant amount of trials to
out-perform even the somewhat loose bounds given by standard martingale concentration inequalities, and can hence only be used in regimes where the amount of trials collected
in the experiment is indeed large. 

\section{Materials and Methods}

Here we present a method for analyzing the \pvalues for Bell experiments that is optimal for large classes of Bell inequalities. This method also applies to event-ready schemes as used in~\cite{Hensen_15}, and can also deal with more complicated forms of event-ready procedures (heralding) in which different states are created in each trial (see Figure~\ref{fig:heralding}). In particular, situations in which we apply a different Bell inequality at each trial depending on which state is generated. Furthermore, we show how to bound the \pvalue\ of all Bell experiments using Bentkus' inequality which is optimal up to a small constant.

Before we can state the concept of a \pvalues more precisely, let us briefly recall the concept of a Bell inequality (see e.g.~\cite{BellSurvey} for an in-depth introduction).
For simplicity, we thereby restrict ourselves to Bell inequalities involving two sites (Alice and Bob), but all our arguments hold analogously for an arbitrary number of sites.
As illustrated in Figure~\ref{fig:gameGeneral}, in a Bell experiment we choose inputs $x$ and $y$ to Alice and Bob, and can record outputs $a$ and $b$~\footnote{Note that 
we here use the more common notation of $a$ and $b$ being outputs, and $x$ and $y$ being inputs. The roles are reversed in~\cite{Hensen_15}}. If the experiment
was governed by a \LHVM, then we could write the probabilities of obtaining outputs $a$ and $b$ given inputs $x$ and $y$ as
\begin{align}
p(a,b|x,y) = \int d\mu(h) p(a|x,h) p(b|y,h)\ ,
\end{align}
where $d\mu$ is an arbitrary measure over hidden-variables $h$, that also include any prior history of the experiment. 
The locality of the model is captured by the fact that $p(a,b|x,y,h) = p(a|x,h) p(b|y,h)$ 
if Alice and Bob are indeed space-like separated. Throughout, we refer to the supplemental material for a formally
precise notation, definitions and derivation.
A Bell inequality then states that for any LHVM
\begin{align}\label{eq:generalBell}
\bellmin \leq \sum_{x,y,a,b} s^{xy}_{ab}\ p(a,b|x,y) \leq \bellmax\ , 
\end{align}
for some numbers $s^{xy}_{ab}$. 
Evidently, in an experiment we never have access to actual probabilities $p(a,b|x,y)$. Nevertheless, Bell inequalities turn out to be very useful to establish bounds on the \pvalues above. 

\begin{figure}
\begin{center}\includegraphics[width=8cm]{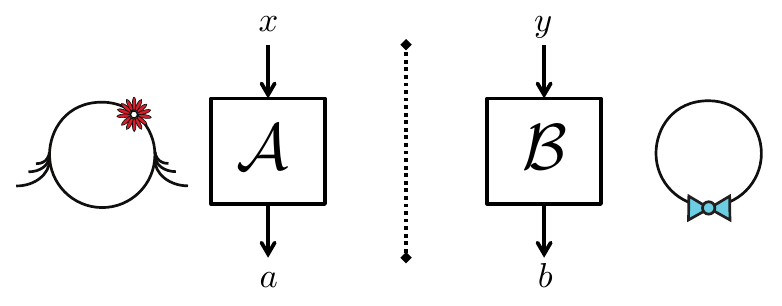}\end{center}
\caption{A Bell test involving two space-like separated sites, labelled Alice and Bob. Alice and Bob receive two randomly chosen inputs $x$ and $y$, and produce outputs $a$ and $b$. We indicate that Alice and Bob are space-like separated via the dotted line.
When testing the CHSH inequality, for example, the inputs and outputs can be taken to be single bits $x,y,a,b \in \{0,1\}$. Viewing CHSH as a non-local game, 
the winning condition is that $x \cdot y = a\oplus b$ (we use the shorthand $a\oplus b$ to denote $a+b\mod 2$). This means that in one trial of the experiment, we check whether $x \cdot y = a \oplus b$ and if yes we increment the number
$c$ of wins by $1$. 
For all Bell inequalities that are win/lose games (see Section~\ref{sec:nonlocal}), we analogously count the number of wins. 
General Bell inequalities (see Section~\ref{sec:general}) can also be cast as a game in which we do not just decide on whether Alice and Bob win or lose, but instead 
assign a score to each correct answer. In the experiment, we then compute the total score from the inputs and outputs observed. 
Our analysis is analogous for Bell inequalities involving more than two sites.
}
\label{fig:gameGeneral}
\end{figure}

Let us now rephrase this inequality in a way that will make our approach more intuitive later on.
In an experiment we choose settings with some probability $p(x,y)$, hence, it will be convenient to define 
\begin{align}
s_{ab|xy} = s^{xy}_{ab}/p(x,y)\ . 
\end{align}
For the moment, let us assume we have perfect random number generators, and that we 
choose the settings $x$ and $y$ uniformly such that $p(x,y) = p(x)p(y)$ where $p(x) = 1/N_x$
and $p(y) = 1/N_y$. The Bell inequality then reads
\begin{align}\label{eq:BellNorm}
\bellmin \leq \frac{1}{N_x N_y} \sum_{x,y,a,b} s_{ab|xy}\ p(a,b|x,y) \leq \bellmax\ .
\end{align}
The reason why this notation is convenient is because we can now think of $s_{ab|xy}$ as a score that Alice and Bob obtain when giving answers $a$ and $b$ for questions $x$ and $y$.
We thus adopt a modern formulation of Bell inequalities in terms of games~\cite{BellSurvey}. The statement that an LHVM governs the experiment then means that Alice and Bob 
can only use a local-hidden variable strategy to achieve a high score in the game. Using this formulation it is clear that the term in~\eqref{eq:BellNorm} is just the average score
that Alice and Bob can hope to achieve in the next trial.  Since the Bell inequality holds for any local-hidden variables, including the history, it is clear that playing the game $n$ times in succession, i.e., performing $n$ trials of the experiment corresponds to a classic example of martingale sequence (see supplemental material).

To analyze the experimental data we then proceed as follows: In trial $j$, we compute the score $s_{a_jb_j|x_jy_j}$ that Alice and Bob obtain for the 
inputs $x$ and $y$ and outputs $a$ and $b$ we observed in that trial. By adding all these numbers we compute the total score $c = \sum_{j=1}^n s_{a_jb_j|x_jy_j}$ after performing $n$ trials. 
The \pvalues then corresponds to 
\begin{align}
&P{\rm -value} \leq \max_{LHVM} \Pr[\mbox{Alice and Bob score C } \geq c\nonumber\\
&\qquad\qquad\qquad\qquad\qquad \mid LHVM]\ .
\end{align}
That is, the probability that Alice and Bob would obtain a score $C$ that is at least as large $C \geq c$ as the score $c$ actually observed in our experiment.

Note that the choice for the score function is not unique. The only restriction, in order to define a $\pvalue$, is that the score needs to be a valid test statistic. 
A test statistic is a function that assigns a real value to each possible experimental outcome. Then, the $\pvalue$ is the probability, under the null hypothesis, that the value of the test statistic is equal or larger to the value obtained from the observed data. There are many possible score functions that verify this restriction, though we would argue that the one used here is particularly natural.

\section{Results}
\subsection{$P$-values for win/lose games}\label{sec:nonlocal}

We first obtain optimal $P$-values for a certain class of Bell inequalities, also known as \emph{non-local games}. 
In particular, this includes the Bell inequalities phrased in terms of correlation functions such as the famous CHSH inequality~\cite{Clauser_69}.
What sets these inequalities apart is that the scores $s_{ab|xy}$ can take on only two values, which we associate with winning or losing the game.

\subsubsection{Winning probability}
To illustrate, how Bell inequalities correspond to games, let us consider the CHSH correlation function
\begin{align}
\langle A_0 B_0 \rangle + \langle A_1 B_0 \rangle + \langle A_0 B_1 \rangle - \langle A_1 B_1 \rangle\ ,
\end{align}
where $A_x$ and $B_y$ correspond to the observables measured by Alice and Bob respectively (see Figure~\ref{fig:gameGeneral}). 
Note that we can write one of the correlators as
\begin{align}
&\langle A_x B_y \rangle = \nonumber\\
&\qquad \sum_{a} p(a,b=a|x,y) - \sum_a p(a,b=a \oplus 1|x,y)\ .
\end{align}
In terms of the score function, this means that $s_{a,b|x,y} = 1$ if $a=b$ and $s_{a,b|x,y} = -1$ if $a \neq b$.
Note that
in any game in which $s_{a,b|x,y}$ can only take on these two values
we can think of the probability that Alice and Bob win for a particular choice of measurement
settings $x$ and $y$ as
\begin{align}
p(\win|x,y) &= \sum_{\substack{a,b\\s_{a,b|x,y} = 1}} p(a,b|x,y)\ ,\\
p(\lose|x,y) &= \sum_{\substack{a,b\\s_{a,b|x,y} = -1}} p(a,b|x,y)\\
& = 1 - p(\win|x,y)\ .
\end{align}
Any Bell inequality for which $s_{a,b|x,y} \in \{\pm 1\}$~\footnote{By normalizing the Bell inequality if needed.}
can thus be written
as 
\begin{align}
&\sum_{x,y} p(x,y) \left(p(\win|x,y) - p(\lose|x,y)\right) =\\
&\qquad \sum_{x,y} p(x,y) 2 p(\win|x,y) - 1 \ .
\end{align}
To draw full analogy with the usual representation of non-local games (see e.g.~\cite{BellSurvey}) 
let us normalize the scores to be $0$ and $1$ instead by defining $\hat{s}_{a,b|x,y} = s_{a,b|x,y}/2 + 1$.
We then have
\begin{align}\label{eq:winningProb}
p(\win) = \sum_{x,y} p(x,y) p(\win|x,y)\ ,
\end{align}
which is precisely the probability that Alice and Bob win the non-local game~\cite{BellSurvey}.
In this language, a Bell inequality now takes on the form
\begin{align}\label{eq:winningBound}
p(\win) \leq \bellwin\, 
\end{align}
where $\bellwin$ denotes the optimal winning probability that
can be achieved using an LHVM.
Note that if necessary, $\bellwin$ can be obtained by normalizing the given values $\bellmin,\bellmax$ appropriately.

\subsubsection{Analyzing data}

The following steps need to be taken to obtain a \pvalues for an experiment based on a non-local game, where for simplicity we first consider schemes that are not event-ready.
We refer to the supplemental material for formal definitions and derivation.

First, we determine a bound on the bias of the random number generator. We will never be able to generate settings $x$ and $y$ exactly according
to the specific distributions $p(x)$ and $p(y)$, instead we will generate the settings according to some other distributions $\tilde p(x)$ and $\tilde p(y)$. We are interested in the numbers $\tau_A$ and $\tau_B$ such that
\begin{align}
|p(x)-\tilde p(x)| &\leq \tau_A\ , \\
|p(y)-\tilde p(y)| &\leq \tau_B\ .
\end{align}
It is clear that for any physical device, these are estimates ideally supported by a theoretical device model with clearly specified assumptions.

Second, we need to obtain a bound on the winning probability using such imperfect random number generators (RNGs).
\begin{align}\label{eq:betaBound}
\tilde{p}(\win|{\rm History}) = \sum_{x,y} \tilde{p}(x,y) p(\win|x,y) \leq \tilde{\beta}_{\win}\ ,
\end{align}
that is valid for all LHVM, where we condition on the history of the experiment. 
Such a bound can be obtained analytically for many inequalities, including CHSH (see supplemental material).
In general, a bound on $\tilde{\beta}_{\win}$ can be computed numerically using a linear program (LP), when re-normalizing the score functions $\hat{s}_{a,b|x,y} \in \{0,1\}$ as above.
We remark that that this LP
has size that is exponential in the number of inputs and outputs, but can nevertheless be solved numerically when these 
are small enough which is typically the case in all experimental Bell tests. It is known that it is NP-hard to compute the
winning probability for {arbitrary non-local games~\cite{Cleve_04}.}

Third, in each of the $n$ experimental trials, we generate inputs $x$ and $y$ and record outputs $a$ and $b$.
In the end, we count the number $c$ of trials in which Alice and Bob won the game, i.e., the number of times
$s_{ab|xy} = 1$.

Finally, we compute the \pvalue. The interpretation of the \pvalues is the probability that Alice and Bob win at least $c$ times, maximized over any LHV strategy.
\begin{align}
&P{\rm -value} = 
\max_{LVHM} \Pr[\mbox{Alice and Bob win} \\
&\qquad \qquad\qquad \qquad \mbox{at least\ } c \mbox{ times }\mid LHVM]\ .\nonumber
\end{align}
As we prove in the supplemental material, for all LHVMs including arbitrary memory effects,
\begin{align}\label{eq:nlp}
P{\rm -value} \leq \sum_{\ellInd=c}^n{n\choose i} \left(\tilde{\beta}_{\win}\right)^i\left(1-\tilde{\beta}_{\win}\right)^{n-i}\ .
\end{align}
This bound is a generalization of~\cite{Barrett_02} and~\cite{Bierhorst_14} that already had given a binomial upper bound for one particular win/lose game, the CHSH game, when the RNGs are perfect, and no event ready-scheme is used.

We emphasize that this bound is tight, whenever~\eqref{eq:betaBound} is tight. 
That is, there exists a LHVM that produces at least $c$ wins with this probability, and this LHVM
does not use any memory. While a theoretical analysis is of course necessary to prove~\eqref{eq:nlp} it follows that the memory loophole~\cite{Barrett_02}
can only be exploited for general Bell inequalities, where it indeed turns out to be significant. Figures~\ref{fig:chsh} and~\ref{fig:mermin} 
illustrate this bound for the CHSH and Mermin's inequality~\cite{Mermin_90}.

\begin{figure}
\begin{center}
\includegraphics[width=8cm]{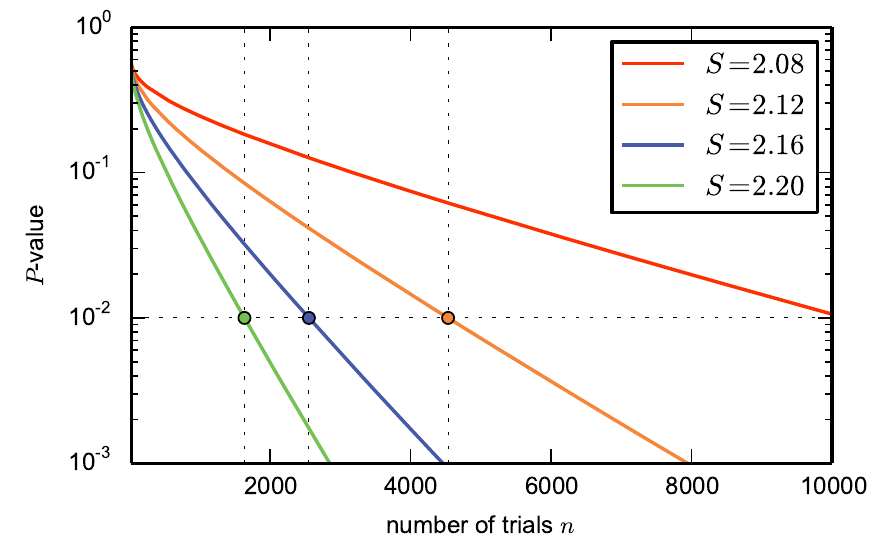}\end{center}
\caption{$P$-values for the CHSH inequality with imperfect random number generators (the bias is $\tau=1.08\cdot 10^{-5}$) in regimes where the violation is very low, but the number of trials is large. The $P$-values are computed with \eqref{eq:HeraldingBound}. The curves show the \pvalue\ as a function of the number of \evs\ for fixed violation values: $S=2.08$, $S=2.12$, $S=2.16$ and $2.20$. The dashed horizontal line is set at \pvalue$=0.01$. This line is crossed at $n=10195$, $n=4534$, $n=2552$ and $n=1635$ \evs\ respectively. }
\label{fig:chsh}
\end{figure}

\begin{figure}
\begin{center}\includegraphics[width=8cm]{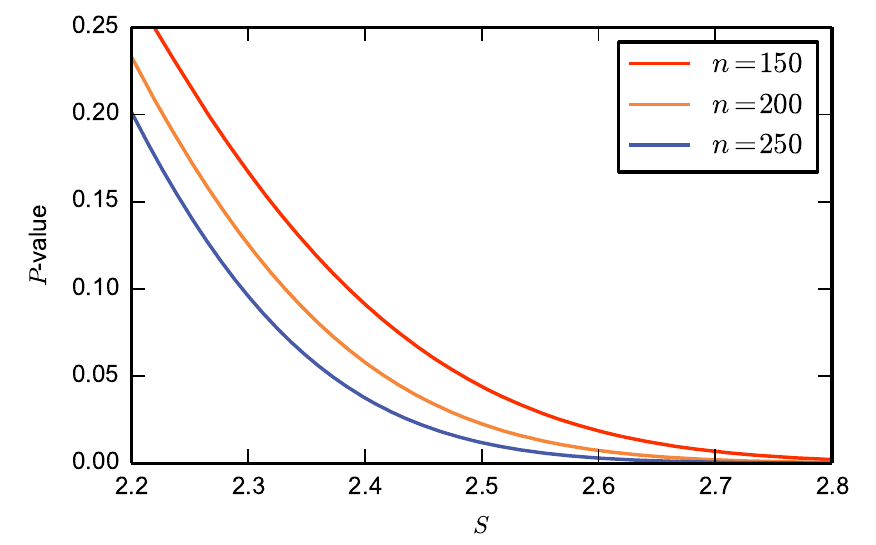}\end{center}
\caption{$P$-values for the Mermin's inequality~\cite{Mermin_90} with perfect random number generators. 
Mermin's inequality is a tripartite inequality in which each party has two inputs and two possible outputs. 
It is an example of a non-bipartite inequality that has already been violated in the laboratory \cite{Pan_00,Zhao_03,Erven_14}. 
The three parties Alice, Bob and Charlie receive three random chosen inputs $x$, $y$ and $z$ with the promise that the parity of the inputs is even, that is that the inputs are limited to $(0,1,1),(1,0,1),(1,1,0),(0,0,0)$, and produce outputs $a$, $b$ and $c$ which can also be taken to be bits. That is:  $x,y,z,a,b,c\in\{0,1\}$. The winning condition for Mermin's inequality is that $a\oplus b\oplus c=x \vee y\vee z$. That is the game is won if the xor of the outputs equals $0$ when $(x,y,z)=(0,0,0)$   
and if the xor of the outputs equals $1$ in the remaining cases. Hence we get: $s_{abc|xyz}=a\oplus b\oplus c\oplus 1$ when $(x,y,z)=(0,0,0)$ and $s_{abc|xyz}=a\oplus b\oplus c$ when $(x,y,z)\neq (0,0,0)$.
The winning probability for this game is $p({\textrm{win}})=3/4$ \cite{Brassard_04}, but note that in contrast with CHSH if Alice, Bob and Charlie share entanglement they can win with probability one. 
The curves show the \pvalue\ as a function of $S=8(c/n-1/2)$ for fixed number of trials $n$ ($c$ is the number of wins). The three curves show from top to bottom the \pvalues for $n=150$, $n=200$ and $n=250$.  
The $P$-values are computed with the binomial upper bound \eqref{eq:nlp}.}
\label{fig:mermin}
\end{figure}

\begin{figure}
\begin{center}\includegraphics[width=8cm]{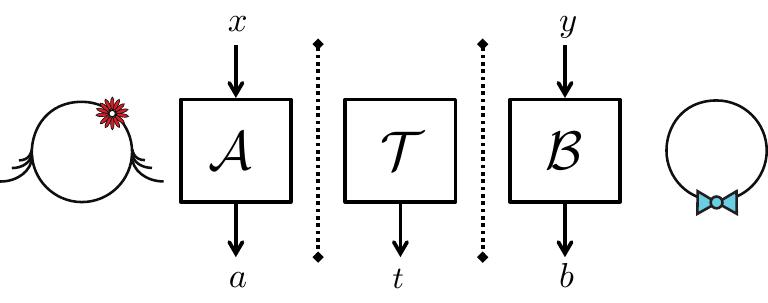}\end{center}
\caption{A Bell test using an event-ready scheme as proposed by Bell~\cite{Bell_81,Bell_04}. In an event-ready scheme, there is an additional site which we call the ``heralding station'' that is space-like separated from Alice and Bob at the time they receive their inputs (see Figure~\ref{fig:gameGeneral}). This heralding station can be under full control of the local-hidden variable model. It takes no input, but produces a tag $t$ as output. In the simplest case, $t$ is just a single bit where $t=1$ corresponds to `yes' and $t=0$ to `no'. If yes, then we check the winning condition for Alice and Bob as in Figure~\ref{fig:gameGeneral}. If no, then no record is made (i.e., the null game is played). 
In physical implementations such as~\cite{Hensen_15} this tag indicates whether an attempt to produce entanglement was successful. 
More complicated scenarios are possible, in which the tag $t$ takes on more than two-values. Depending on $t$, a particular game is played, i.e., scores are computed
according as dictated by the game labelled by $t$. In physical implementations this is interesting when two different entangled states can be created in the event-ready scheme, and
each state is best for a particular game. An example is given by CHSH, where different Bell states are created and we play two different CHSH games with $x \cdot y = a \oplus b$ or $x \cdot y = a \oplus b \oplus 1$. Using both states can improve the time scales at which statistical confidence can be obtained.}
\label{fig:heralding}
\end{figure}

\subsubsection{Event-ready schemes}\label{sec:evready}

To illustrate the analysis of event-ready schemes, let us here focus on the usual case where the tag (see Figure~\ref{fig:heralding}) can be either $t=0$ (null game, no entanglement was made) or $t=1$ (one game, one specific state was made). We will use the term attempt to refer to an attempt to create entanglement (outcome $t=0$ or $t=1$) and reserve the word trial for those in which $t=1$. In the supplemental material, we will discuss more complex versions of event-ready schemes in which different entangled states can be created, and we employ a different game for each state. 

While it is important that the random numbers are chosen independently of the tag $t$, we otherwise allow the LHVM arbitrary control over the statistics of heralding station. In particular, this means that the LHVM may use more (or less) attempts to realize $c$ wins on $n$ $\evs$ than we actually observed during the experiment. 

Specifically, 
\begin{align}
P{\rm -value} &= \max_{LHVM} \sum_{\substack{t^m\in\{0,1\}^m\\|t^m|=n}} \Pr[t^m\mid LHVM]\nonumber\\
&\qquad \qquad \Pr[C \geq c {\rm\ wins}\mid LHVM, t^m]\ ,
\end{align}
where $t^m = t_1,\ldots,t_m$, $|t^m|$ denotes the number of ones in $t^m$ and the maximization over LHVM includes a maximization over an arbitrarily large
number of attempts $m$ and heralding statistics. 
As we will formally show in the supplemental material, 
\begin{align}
\label{eq:HeraldingBound}
P{\rm -value} = \sum_{\ellInd=c}^n{n\choose i} \left(\tilde{\beta}_{\win}\right)^i\left(1-\tilde{\beta}_{\win}\right)^{n-i}\ .
\end{align}
That is, we can formally ignore the non successful \TTs. The \pvalue\ only depends on the \evs. 

\subsection{General games}\label{sec:general}

Let us now move on to considering general games, that is, games in which the score functions $s_{ab|xy}$ take
on more than two possible values. 
As before, we first need to consider the bias. 
Our bound will depend on the values of 
\begin{align}
s_{\rm max} &= \max_{a,b,x,y} s_{ab|xy}\ ,\\
s_{\rm min} &= \min_{a,b,x,y} s_{ab|xy}\ .
\end{align}
Recall that since $s_{ab|xy} = s^{xy}_{ab}/p(x,y)$ the distribution $p(x,y)$ and hence also the bias influence $s_{\rm max}$ and $s_{\rm min}$.
Second, we again compute the total score
\begin{align}
c = \sum_{j=1}^n s_{a_jb_j|x_jy_j}\ ,
\end{align}
where $x_j$, $y_j$, $b_j$ and $a_j$ are the inputs and outputs used during trial $j$ respectively. 
We then have that 
\begin{align}
P{\rm -value} = \max_{LHVM} \Pr[C \geq c|LHVM]\ ,
\end{align}
where $C$ is the random variable corresponding to obtaining a particular score using the LHVM strategy. 
Using the Bentkus' inequality, we prove in the supplemental material that
\begin{align}\label{eq:GeneralBound}
P&{\rm -value} \leq \nonumber\\
          &e\Bigg(\left( \sum_{\ellInd=\lfloor \delta\rfloor}^n{n\choose i} \left(\hat\gamma\right)^i\left(1-\hat\gamma\right)^{n-i}\right)^{1-\delta +\lfloor \delta\rfloor}\nonumber\\
          &\left(\sum_{\ellInd=\lceil \delta\rceil}^n{n\choose i} \left(\hat\gamma\right)^i\left(1-\hat\gamma\right)^{n-i}\right)^{\delta-\lfloor \delta \rfloor}\Bigg)
\end{align}
where
\begin{align}
\delta &= \sum_{i = 1}^{n} \frac{c_i - s_{\rm min}}{s_{\rm max} - s_{\rm min}}\ ,\\
\hat \gamma &= \frac{\beta_{\rm max} - s_{\rm min}}{s_{\rm max} - s_{\rm min}}\ .
\end{align}

Whenever the Bell inequality is normalized such that $s_{\rm min} = 0$ and $s_{\rm max} = 1$ this becomes
\begin{align}\label{eq:GeneralBoundNorm}
P&{\rm -value} \leq \nonumber\\
          &e\Bigg(\left( \sum_{\ellInd=\lfloor c\rfloor}^n{n\choose i} \left(\beta_{\max}\right)^i\left(1-\beta_{\max}\right)^{n-i}\right)^{1-c+\lfloor c\rfloor}\nonumber\\
          &\left(\sum_{\ellInd=\lceil c\rceil}^n{n\choose i} \left(\beta_{\max}\right)^i\left(1-{\beta}_{\max}\right)^{n-i}\right)^{c-\lfloor c\rfloor}\Bigg)
\end{align}
where $\lfloor c\rfloor$ and $\lceil c\rceil$ stand respectively for the greatest integer smaller than $c$ and the smallest integer larger than $c$.

If we treat a win/lose game as a general game we can also upper bound the $\pvalue$ by \eqref{eq:GeneralBoundNorm}. However, if we compare this formula with \eqref{eq:nlp}, we see that we have lost a factor of $e$. We have obtained a simple formula that can address general games but it is not tight. It remains unknown whether or not $e$ is the optimal prefactor, but it is known that for general games it cannot be smaller than 2 \cite{Bentkus_04}.

In some cases it is possible to transform a general game into a win/lose game by postselecting the trials that take the maximum and minimum value \cite{Gill_03b,Bierhorst_15}. In that situation, it would be possible to apply the tight bounds for win/lose games. Techniques sometimes referred to as ``speeding up time''~\cite{Gill_03b,Kofler_14} can analogously be used in conjunction with this refined bound. 

The idea behind this bound is to model an experiment as a bounded difference supermartingale, where we note that
a Bell inequality is nothing else than the expectation of the score random variable $C_j$ in trial $j$ conditioned
on the history leading up to that trial. That is,
\begin{align}\label{eq:expected}
\beta_{\rm min} \leq \mathbb{E}[C_j|{\rm History}] \leq \beta_{\rm max}\ ,
\end{align}
where the expectation is taken over all inputs $x$, $y$ and outputs $a$ and $b$. 
A (super)martingale is a concept known from probability theory (see supplementary material for details). 
A sequence $M_1,M_2,\ldots$ of random variables is known as a supermartingale, if the expectation value
of the difference $M_n - M_{n-1}$ conditioned on the history is always negative. Choosing $M_j$ to be a weighted
sum of the differences $\sum_{j=1}^n C_j - \beta_{\rm max}$ one can easily obtain such a Martingale.
The key aspect of a Martingale is that even though the subsequent variables are not independent from each other, 
we nevertheless observe a concentration akin to the law of large numbers for processes which are independent from each other.
The prime example is tossing a coin $n$ times. Indeed, thinking of ``heads'' as ``win'' and ``tails'' as ``loose'', we can easily evaluate the probability that
we get ``win'' more than $k$ times. When a process is a Martingale a similar argument holds, even if the coin can take many values and depend on the history.

Several other martingale bounds have been used in the past. We have chosen as examples McDiarmid's   
inequality~\cite{Mcdiarmid_89} as given in~\cite{Zhang_13}
\begin{align}
P{\rm -value} &\leq \bigg(\left(\frac{s_{\rm max} - \beta_{\rm max}}{s_{\rm max}-c/n}\right)^{\frac{s_{\rm max}-c/n}{s_{\rm max}-s_{\rm min}}}\nonumber\\
                   &\qquad                     \left(\frac{\beta_{\rm max}-s_{\rm min}}{c/n-s_{\rm min}}\right)^{\frac{c/n-s_{\rm min}}{s_{\rm max}-s_{\rm min}}}\bigg)^n
                    \label{eq:mcdiarmid}
\end{align}
and  Azuma-Hoeffding as used in~\cite{Pironio_10}
\begin{equation}
P{\rm -value} \leq \exp \left(-n\frac{(c/n-\beta_{\rm max})^2}{2d^2}\right)
\label{eq:azuma}
\end{equation}
where $d=\max\{\beta_{\rm max} - s_{\rm min},s_{\rm min}-\beta_{\rm min}\}$.

We provide an example of the application of these three bounds for the Collins-Gisin-Linden-Massar-Popescu (CGLMP) inequality~\cite{Collins_02} in Figure~\ref{fig:cglmp}. 

\begin{figure}
\begin{center}\includegraphics[width=8cm]{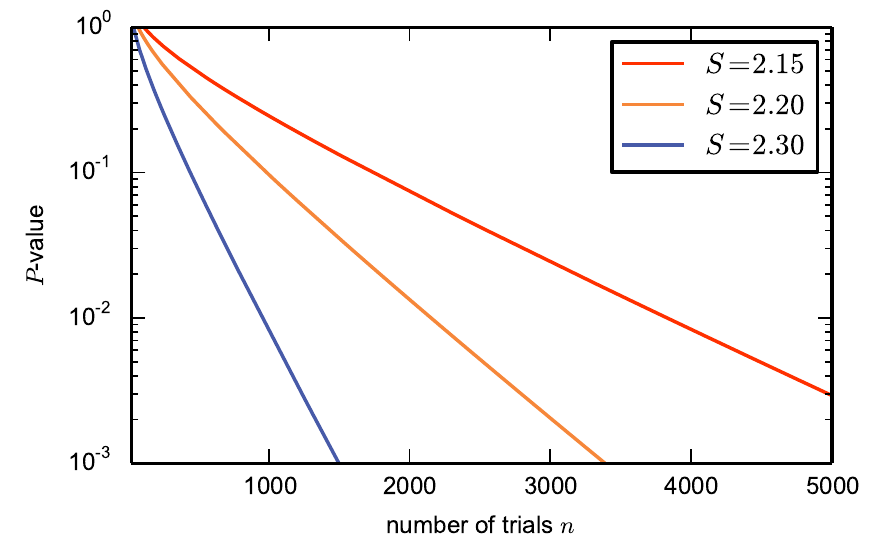}\end{center}
\caption{$P$-values for CGLMP's inequality \cite{Collins_02} with perfect random number generators. CGLMP is a sequence of bipartite inequalities in which each party has two inputs and $d\geq 2$ possible outputs. This is an example of a general game within experimental reach \cite{Dada_11}. 
The inequality is $\sum_{k=0}^{\lfloor d/2-1\rfloor}\sum_{a,x,y}\left(1-\frac{2k}{d-1}\right)(p(a,a+k+xy|x,y)-p(a,a-k-1+xy|x,y))\leq 2$. Let $k\in 0\ldots \lfloor d/2-1\rfloor$, we can extract from the inequality the score functions: $s_{ab|xy}=4\left(1-\frac{2k}{d-1}\right)$ if $b=a+k+xy$, $s_{ab|xy}=-4\left(1-\frac{2k}{d-1}\right)$ if $b=a-k-1+xy$ and $s_{ab|xy}=0$ in the remaining cases.  
The three curves show the \pvalue\ as a function of the number of \TTs\ for a fixed average score $S=\frac{1}{n}\sum_{j=1}^ns_{a_jb_j|x_jy_j}$. 
From top to bottom the curves show the \pvalues for $S=2.15$, $S=2.20$ and $S=2.30$.  
The $P$-values are computed via Bentkus' inequality \eqref{eq:GeneralBound}.}
\label{fig:cglmp}
\end{figure}

\begin{figure}
\begin{center}\includegraphics[width=8cm]{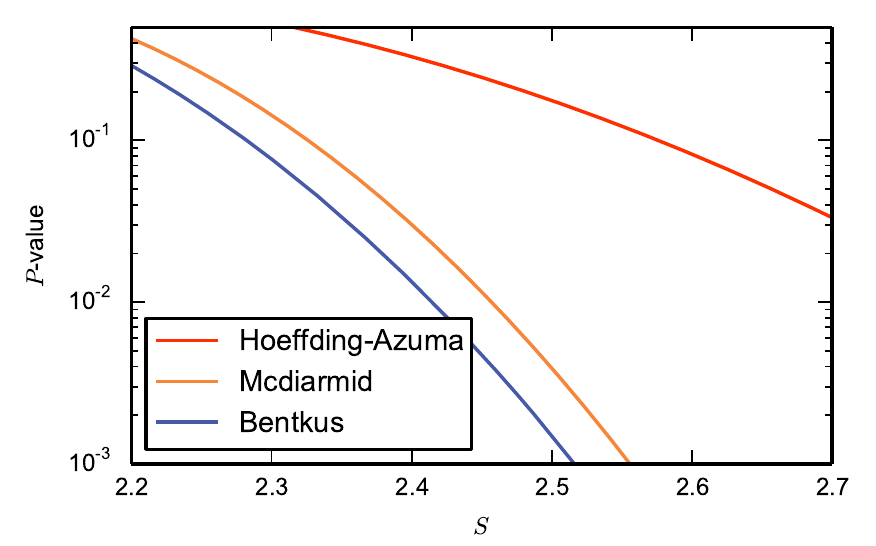}\end{center}
\caption{$P$-values for CGLMP's inequality \cite{Collins_02} with perfect random number generators. From top to bottom the curves show the \pvalues for $n=500$ $\evs$ computed with Azuma-Hoeffding, Mcdiarmid and Bentkus' inequalities. The inequalities are given by  \eqref{eq:azuma},  \eqref{eq:mcdiarmid} and  \eqref{eq:GeneralBound}.}
\label{fig:cglmpcompare}
\end{figure}

\section{Discussion}

\subsection{Before conducting the experiment}

To ensure sound data collection, there are several important considerations to make before the experiment takes place.
These are standard in statistical testing, and in essence say that the rules on how the statistical analysis is performed is decided independent of the data. 
This can be achieved by establishing those rules before the data collection starts. 
First, we choose a Bell inequality. Not all Bell inequalities lead to the same statistical confidence. In Section~\ref{sec:chooseBell} we discuss methods for obtaining a good one. While there may be future analyses that allow a partial optimization over Bell inequalities using the actual experimental data, we emphasize
that the procedure above assumes a fixed inequality has been chosen ahead of time. Second, there are two ways to deal with imperfect random number generators, and a choice should be made as discussed in Section~\ref{sec:chooseRNG}. 
Third, we assume that the number of trials to be collected is independent of the data. 
This means that we do not decide to take another few trials if the \pvalues is not yet low enough for our liking, a practise also known as \pvalues
fishing in statistics. There are ways to augment the analysis~\cite{stopping} to safely collect more data in some specific instances, but this brings many subtleties. 
A number of $\evs$ $n$ can be determined from the expected violation given prior device characterization, aiming for a particular \pvalue. 

\subsection{How to deal with imperfect RNGs}\label{sec:chooseRNG}

From the discussions above, it becomes clear that there are two ways to deal with imperfect RNGs. The first is of interest in win/lose games. If there is a bias $\tau$, then the winning probability~\eqref{eq:winningProb} simply increases. This means that when we perform an experiment based on a win/lose game in which we use an imperfect RNG, then the game remains win/lose and the bound in~\eqref{eq:nlp} still applies. Since this is a simple Binomial distribution, without any additional factor $e$ this is desirable if the bias is small. 

However, we saw from the analysis of general games that there is a second way. When considering a general Bell inequality~\eqref{eq:generalBell}, we make no statements about the probabilities of choosing settings $p(x,y) = p(x)p(y)$. 
Starting from a scoring function $s^{xy}_{ab}$ we can define $s_{ab|xy} = s^{xy}_{ab}/p(x,y)$ to introduce an explicit
dependence on the input distribution $p(x,y)$ of our choosing. 
Using RNGs with a bias then merely affects $p(x,y)$ and thus the maxima and minima of the scoring
functions $s_{ab|xy}$ which enter into the bound given in~\eqref{eq:GeneralBound}. 
It is crucial to note that when defining $s_{ab|xy}$ as above, then a win/lose game in which $p(x,y)$ was chosen 
to be perfectly uniform, can now turn into a general game. That is, we will no longer have that the scoring functions 
$s_{ab|xy}$ take on only two values. This means that we have to use the general bound~\eqref{eq:GeneralBound} carrying
the additional factor $e$, as opposed to~\eqref{eq:nlp}. 

How we deal with imperfect RNGs thus depends: if we start with a win/lose game, and if the bias is small, then it is typically advantageous to preserve the win/lose property of the game and derive a new winning probability as a function of the bias. If, however, the bias is very large, then it can be advantageous to sacrifice the win/lose property, and adopt the analysis for general games.
If the game was not win/lose to begin with, we always adopt the second method.

\subsection{Selecting a Bell inequality}\label{sec:chooseBell}
One of the main objectives of a Bell experiment is to quantify the evidence against a $\LHVM\ $, hence ideally one would like to choose a game that would yield the lowest $\pvalue\ $ for a fixed number of $\evs$. The optimization of games with this objective is a non-trivial task. A reasonable alternative which one can use as heuristic is to mazimize the 
gap between the expected score achievable in the experiment and the expected score that a $\LHVM\ $ can attain.  
In other words, we are looking for a Bell inequality for which the violation we can observe is as large as possible. To find such an inequality, standard linear programming methods can be used (see e.g.~\cite{BellSurvey}). 

To apply them we assume that a reasonably good guess is available as to what the probabilities $p(a,b|x,y)$ are in the experiment. Such a guess can be made by either analyzing data collected prior to the Bell experiment and approximating probabilities by relative frequencies, or by having sufficient confidence in the theoretical model that describes the experiment and calculating the probabilities from this model. 

Suppose that in the estimation process we find some estimates of $p(a,b|x,y)$. If such probabilities could be realized by a LHVM, then we could write them as a mixture of deterministic local strategies. To make this precise, let 
$\lambda=(a_1,\ldots,a_{|X|},b_1,\ldots,b_{|Y|})$ denote a deterministic strategy in which Alice and Bob give outputs $a_x$ and $b_y$ for inputs $x = \{1,\ldots,|X|\}$ 
and $y \in \{1,\ldots,|Y|\}$. In terms of a probability distribution, this would correspond to a distribution $d_{\lambda}(a,b|x,y)$ such that $d_{\lambda}(a,b|x,y) = 1$ if and only if $a = a_x$ and $b = b_y$ as indicated by the vector $\lambda$, and $d_{\lambda}(a,b|x,y) = 0$ otherwise.  A behaviour, that is distributions $p(a,b|x,y)$, is local if and only if 
\begin{align}\label{eq:ppDef}
p(a,b|x,y) = \sum_{\lambda} q_{\lambda}\ d_{\lambda}(a,b|x,y)\, 
\end{align}
where the sum is taken over all $|X|^{|A|} |Y|^{|B|}$ possible $\lambda$~\cite{BellSurvey}, where $|A|$ and $|B|$ denote the number of possible outputs for Alice and Bob, 
and 
\begin{align}
\forall \lambda, q_{\lambda} > 0, \mbox{\ and\ } \sum_{\lambda} q_{\lambda} = 1\ .  
\end{align}
We note that one can test whether such $q_{\lambda}$ exist, i.e. whether the behaviour is local, using a linear program~\cite{Zukowski_99,Kaszlikowski_00}.  
The dual of this linear program can be used to find a Bell inequality that certifies a behaviour $p(a,b|x,y)$ is \emph{not} local~\cite{BellSurvey}. One can easily adapt this linear program to search for an inequality that achieves a high violation. Specifically, 
\begin{align}
\textrm{maximize} & {\rm\ Violation} = \sum_{x,y,a,b} s^{xy}_{ab}\ p(a,b|x,y) - S\nonumber\\
&\sum_{x,y,a,b} s^{xy}_{ab}\ d_{\lambda}(a,b|x,y) \leq S, \forall \lambda\\
&0 \leq s^{xy}_{ab} \leq 1, \forall x,y,a,b\ ,
\end{align}
where the $p(a,b|x,y)$ and $d_{\lambda}(a,b|x,y)$ are givens, and we optimize over $s^{ab}_{xy}$ (see~\cite{BellSurvey} for details).
Note that the second constraint means that for every LHVM, we have a Bell inequality in which $\beta_{\rm max} = S$
and the difference $V$ is precisely the violation we achieve when 
normalizing the score functions to lie in the interval $[0,1]$ which can be done without loss of generality. 

It is clear from the discussion above that it can be to our advantage to search for a win/lose game, rather than a general game
since the $p$-values for such games are sharper. This can be done by optimizing over score functions in which $s^{xy}_{ab} \in \{0,1\}$. This, however, is now an integer program~\cite{Faigle_13} rather than a linear program~\cite{boyd:book}, which are in general NP-hard to solve~\cite{Faigle_13}. Nevertheless, this may be feasible for the small number of inputs and outputs used in any experimental implementation, and heuristic methods exist.

\subsection{Combining independent experiments}
Suppose that a series of $n$ experiments is run independently. Each experiment could correspond to completely different settings, Bell inequalities, etc. Associated with each experiment we obtain a series of $P$-values corresponding to the probability that each of them was governed by a \LHVM: $(p_i)_{i=1}^n$. In this situation, it is possible to take all the $P$-values associated with each one of the individual experiments and obtain a combined \pvalue. One such a method is Fisher's method \cite{Fisher_25,Elston_91}. With Fisher's method the combined \pvalues is given by the tail probability of $\chi^2_{2n}$, a chi-squared distribution with $2n$ degrees of freedom:
\begin{align}
\label{eq:fisher}
P\textrm{-value} &= \pr{\chi^2_{2n}\geq -2\log \prod_{i=1}^np_i}
\end{align}
The right hand side of this equation can be easily evaluated numerically. However, it can be shown that the tail probability of $\chi^2_{2n}$ accepts the following closed expression:
\begin{align}
\pr{\chi^2_{2n}\geq 2x}=x\sum_{i=0}^{n-1}\frac{x^i}{i!}
\end{align}
where we can choose $x=-\log\prod_{i=1}^n p_i$.

However, we make no claim of optimality regarding the combined \pvalue. There is a rich literature on methods for combining $P$-values \cite{Loughin_04} and depending on the concrete situation a different choice should be made.

\subsection{Conclusions}

We have shown how to derive (nearly) optimal $P$-values for all Bell inequalities that can easily be applied to evaluate the data collected in experiments. A suitable Bell inequality can be found as outlined above, however, it would be interesting to combine this method with the
numerical search for inequalities in~\cite{Zhang_10, Zhang_11, Zhang_13}. The latter can adaptively find the best way to discriminate between LHVMs and theories like quantum mechanics that go beyond local-hidden variables that is asymptotically optimal, but requires a significant amount of data to train. 

We note that there exist many ways to extend the methods presented here to deal with specific situations at hand, for example, by conducting multiple experiments in succession and using data from prior runs to find more suitable Bell inequalities in the next instance. 

We emphasize that the methods outlined here can be used to test other models than LHVMs. It is clear from the proof 
that only the winning probability in~\eqref{eq:winningBound}, or the expectation value~\eqref{eq:expected}, depends on the model to be tested. The argument that extends these bounds for a single trial to a bound on the $P$-value for the entire experiment allowing arbitrary memory in the devices, however, does not depend on the model tested.
In particular, this means that any theories that predict bounds of the form~\eqref{eq:winningBound} and~\eqref{eq:expected} 
are excluded with the same bound on $P$-value. This also makes it apparent how one can extend the analysis to refute models that are more powerful than an \LHVM. For example, Hall~\cite{hall} defined and quantified interesting relaxations of an \LHVM, with reduced free will, or where some amount of signalling is allowed. It is straightforward to adapt the analysis of~\cite{hall} to derive bounds on~\eqref{eq:winningBound} and~\eqref{eq:expected} 
to subsequently obtain a $P$-value for testing such extended models. Note that since
Alice and Bob obtain an advantage by allowing models such as~\cite{hall}, i.e. they are allowed more powerful strategies, hence they can achieve a higher score in the game. This implies that concrete scores will result in higher $P$-values and lower confidence.

Furthermore, while we focused the discussion on tests of Bell inequalities, our methods can also be applied to the study of certified 
randomness as in~\cite{Colbeck_07, Pironio_10,Pironio_13}, or more generally to tests of e.g. non-contextual models that can be phrased as one player games.

\acknowledgements
We thank Hannes Bernien, Peter Bierhorst, Andrew Doherty, Ana{\"i}s Dr{\'e}au, Richard Gill, Peter Gr{\"u}nwald, Ronald Hanson, Bas Hensen, Jed Kaniewski,
Laura Man{\v c}inska, Corsin Pfister, Tim Taminiau, Thomas Vidick and Yanbao Zhang for discussions and/or comments
on an earlier version of this manuscript. We also thank the referees for their careful reading and suggestions. 
DE and SW are supported by STW, Netherlands and an NWO VIDI Grant.

\onecolumngrid
\appendix

\bigskip

In this supplemental material, we formalize and prove our claims. To accomplish this, we first need to introduce more precise notation and a formal description of LHVMs in Section~\ref{sec:notation}. We then proceed to analyze win/lose games in Section~\ref{sec:winLoseProofs}, 
and general games in Section~\ref{sec:General}. 
 
\section{Preliminaries}
\label{sec:notation}

\subsection{Notation}
We will use capital letters to denote random variables: $A,B,\ldots$ and the corresponding lower case letters to denote the value that the random variable takes: $a,b,\ldots$. Instead of the notation $p(a)$ used in the main text, 
we will use the more precise form $\pr{A=a} = p(a)$. During the experiment, we perform many attempts to generate entanglement
in which we will record the inputs $a_\ellInd, b_\ellInd$, outputs $x_\ellInd$ and $y_\ellInd$, and event-ready tags $t_\ellInd$ in each attempt. We will reserve the word trial for the attempts in which $t_\ellInd \neq 0$. Note that in an experiment that does not use an event-ready procedure we always have $t_\ellInd \neq 0$. 

While we restrict our explanations to the bipartite case, we emphasize that is straightforward to extend our analysis to any number of sites and we provide a simple example of how this is done in Figure~4 in the main text. 
A single \TT\ of a (bipartite) game is characterized by the inputs that we denote by $X$ and $Y$ and the corresponding outcomes $A$ and $B$. In the case of an event-ready scheme, the outcome of the event-ready station would be denoted by $T$. Let $\Delta_\ellInd^{a,b,x,y,t}$ be an indicator function
\begin{equation}
\label{eq:bierhorst}
\Delta_\ellInd^{a,b,x,y,t}=\mathds 1\{A_\ellInd=a,B_\ellInd=b,X_\ellInd=x,Y_\ellInd=y, T_\ellInd = t\}\ .
\end{equation}
That is, $\mathds 1\{A_\ellInd=a,B_\ellInd=b,X_\ellInd=x,Y_\ellInd=y,T_\ellInd = t\}$ is itself a random variable that is a function of the random variables $A_\ellInd$,$B_\ellInd$,$X_\ellInd$,$Y_\ellInd$ and $T_\ellInd$.
It takes on the value $1$ if all equalities are satisfied for a particular choice of $a,b,x,y,t$, and $0$ otherwise.

We will let the random variable $C_\ellInd$ stand for the score of the game obtained in trial $\ellInd$. This variable is defined as function of the coefficients $s_{ab|xy}$ 
that characterize the game
\TT\ as
\begin{align}\label{eq:CDef}
C_\ellInd = \sum_{a,b,x,y,t} \Delta_\ellInd^{a,b,x,y,t}\cdot s_{ab|xyt}\ .
\end{align}
Depending on the event-ready tag $t$, a different game might be played, which implies that $s_{ab|xyt}$ depend on $t$. Whenever 
the dependence on $t$ is clear, we will drop $t$ and simply write $s_{ab|xy}$
to avoid cluttering the notation.
The concrete instance of the $\ellInd$-th \TT\ we denote by
\begin{align}
c_\ellInd = c_\ellInd(a,b,x,y,t)=s_{ab|xyt}\ .
\end{align}
We will often drop the dependence on $(a,b,x,y,t)$ by writing $c_\ellInd$
in order to lighten the notation.
The term concrete instance means that $c_\ellInd$ can be computed from the observed data. In an experiment consisting of $m$ attempts, we first need to compute the following number, which is the total score Alice and Bob obtain
\begin{align}\label{eq:ToCompute}
c = 
\sum_{\ellInd = 1}^{m} c_\ellInd(a_\ellInd,b_\ellInd,x_\ellInd,y_\ellInd,t_\ellInd) = 
\sum_{\ellInd = 1}^{m} s_{a_\ellInd b_\ellInd|x_\ellInd y_\ellInd t_\ellInd}\ .
\end{align}
The corresponding random variable is
\begin{align}
C = \sum_{\ellInd = 1}^{m} C_{\ellInd}\ .
\end{align}
It will furthermore be convenient to define the following shorthand
\begin{align}
P_{n,k}(\mathbb{B}_\gamma) = \sum_{i=k}^n {n\choose i} \gamma^i \left(1-\gamma\right)^{n-i}\ .
\end{align}

\subsection{Local-hidden variable models}
In order to formally state the null hypothesis, we briefly need to state what LHVMs are more precisely. More details can be found
in e.g.~\cite{BellSurvey} and~\cite{Bierhorst_15}.
To do so, let us introduce the following sequences of random variables in correspondence with the concrete instances of inputs and outputs denoted by the lowercase letters above.
Let $\mathbf{X}^{\N} = (X_\ellInd)_{\ellInd=1}^\N,\mathbf{Y}^{\N} = (Y_\ellInd)_{\ellInd=1}^\N$ denote the inputs to the boxes where $\ellInd$ is used to label the $\ellInd$-th element,
$\mathbf{A}^{\N} = (A_\ellInd)_{\ellInd=1}^\N,\mathbf{B}^{\N} = (B_\ellInd)_{\ellInd=1}^\N$ the outputs of the boxes,  
 $\mathbf{\HHTot}^{\N} = (\HHTot_\ellInd)_{\ellInd=1}^\N$ the histories of attempts previous to the $\ellInd$-th attempt, $\mathbf{C}^{\N} = (C_\ellInd)_{\ellInd=1}^\N$ denotes the scores at each \TT\ and $\mathbf{T}^{\N} = (T_\ellInd)_{\ellInd = 1}^{\N}$ is the sequence of event-ready signals in the case of an event-ready experiment.  
In an event-ready experiment, we make no assumptions regarding the statistics of the event-ready station, which may be under full control of the \LHVM, and can depend arbitrarily on the history of the experiment. 

The random variable $\HHTot_\ellInd$ models the state of the experiment prior to the measurement. As such, $\HHTot_\ellInd$ includes 
any hidden variables, sometimes denoted using the letter $\lambda$ \cite{BellSurvey}. It also includes 
the history of all possible configurations of inputs and outputs of the prior \TTs\ $(X_j, Y_j, A_j, B_j, T_j)_{j=1}^{\ellInd-1}$.
However, this is the only requirement for $\HHTot_\ellInd$; that is the history may also include other aspects of the experiment. For simplicity, we assume that it is a countable random variable, though in full generality it could be defined over an arbitrary probability space. 

The null hypothesis (to be refuted) is that our experimental setup can be modeled using a \LHVM\ (see~\cite{Bierhorst_14} for more details). This model has the following properties:

\begin{enumerate}
\item \emph{Local randomness generation.}
Conditioned on the history of the experiment the inputs $X_\ellInd,Y_\ellInd$ are independent of each other
\begin{equation}
\forall \ellInd, X_\ellInd \indep Y_\ellInd\mid H_\ellInd\ ,
\end{equation}
and of the output of the event-ready signal
\begin{equation}
\forall \ellInd, X_\ellInd \indep T_\ellInd, Y_\ellInd \indep T_\ellInd \mid H_\ellInd\ .
\end{equation}
We allow $X_\ellInd$ and $Y_\ellInd$ to be partially predictable given the history of the experiment. We use $p_x$ and $p_y$ for the distribution that we are hoping to achieve using imperfect RNGs. However, we assume this target distribution to be the same for all $\ellInd$.
\begin{align}
\forall (\ellInd,x_\ellInd,h_\ellInd), p_x - \tau_A \leq \pr{X_\ellInd = x_\ellInd|H_\ellInd = h_\ellInd} &\leq p_x + \tau_A\ , \label{eq:AliceBias}\\
\forall (\ellInd,y_\ellInd,h_\ellInd), p_y - \tau_B \leq \pr{Y_\ellInd=y_\ellInd|H_\ellInd = h_\ellInd} &\leq p_y + \tau_B\ ,\label{eq:BobBias}
\end{align}
where we define $\tau = \max\{\tau_A,\tau_B\}$.
\item \emph{Locality.}
The outputs $a_i$ and $b_i$ only depend on the local input settings and history: they are independent of each other and of the input setting at the other 
side, conditioned on the previous history and the current \heraldingEvent
\begin{equation}
\label{eq:locality}
\forall \ellInd, (X_{\ellInd},A_\ellInd) \indep (Y_{\ellInd},B_\ellInd)|\HH_\ellInd, T_\ellInd\ .
\end{equation}
\item \emph{Sequentiality of the experiments.}
Every one of the $\NNN$ \TTs\ takes place sequentially such that any possible signalling between different \TTs\ beyond the previous conditions is prevented. The reason for this condition is that this signalling opens the simultaneous measurement loophole \cite{Barrett_02}. Also note, that if the sequentiality condition is not met the history random variable becomes ill-defined. 
\end{enumerate}
Except for these properties the variables might be correlated in any possible way. 

A model that verifies these properties or constraints is what we call an $\LHVM$ and it is under these conditions that our statements on the $\pvalue$ do hold. Then, a small $\pvalue$ can be used to reject the hypothesis that the experiment was governed by such an $\LHVM$. 
Note that some of these constraints are naturally backed by some experimental setups while in some others they might be less justified. For instance, one might argue that that $X_i$ and $T_i$ are independent given the history because the corresponding stations are space-like separated. If, on the other hand, the stations can signal to each other, one might still make the assumption that $X_i$ and $T_i$ are independent. However, in contrast to the scenario in which the stations are space-like separated, a small $\pvalue$ and consequent rejection of the null hypothesis would still leave the door open to other local models that can explain the observed data with high probability, e.g. models in which $X_i$ and $T_i$ are not independent.

\section{Analysis of win/lose games}\label{sec:winLoseProofs}
In this section we consider games where the scoring variable takes only two values. As argued in the main text, these games can always be transformed, via normalization, into games that take the values 1 and 0. We identify these values with winning and losing. This analysis is an extension of the one done for the Delft experiment 
(Supplementary information~\cite{Hensen_15}).
For a win/lose game, the probability of winning in a given \ev\  equals the probability that the score takes the value 1: $\pr{C_\ellInd=1}$. Note that the score $c$ we compute from the data given in~\eqref{eq:ToCompute}
is now just the number of times that Alice and Bob win the game.

Suppose that we perform a win/lose game $n$ times and we observe $c$ wins. The \pvalue\ for an experiment that employs an event-ready procedure with two outputs $t_j = 0$ (no, not ready) and $t_j=1$ (yes, ready) is 
\begin{equation}
\label{eq:pvaluewinlose}
P{\rm -value} \leq \max_{LHVM} \sum_{\substack{t^m \in \{0,1\}^m \\ |t^m|=n}} \pr{T^m = t^m\mid LHVM} \pr{C \geq c\mid LHVM, T^m = t^m}\ .
\end{equation}
Let us now show how to obtain a tight upper bound on \eqref{eq:pvaluewinlose} for all win/lose games in a systematic way. 
We detail the procedure in the following. 

\subsection{Step 1: Bounding the probability of winning the next trial}
 
First, we need to prove that if the experiment is ruled by an \LHVM, then the winning probability of the next trial 
is bounded from above by some $\bellwin^t$ for any possible history $h_\ellInd$ and event-ready signal $t$. 
Such bounds can be obtained in two ways. If a tight bound is achieved for $\bellwin^t$, then our final bound will also be tight and attained by a LHVM strategy that 
does not use any memory.

\subsubsection{A numerical bound using linear programming}

In general, it is always possible to obtain a numerical bound via
a linear program (LP, see e.g.~\cite{BellSurvey} and also Section~IV in the main text). 
While the history can be arbitrary, note that the history can always be reflected in terms of a choice of hidden-variables. 
It is known that these can be taken to be finite, allowing us to compute
\begin{align}
\bellwin^t \leq {\rm maximize} & \sum_{x,y} 
\sum_{a,b} \hat{s}_{ab|xyt} \pr{A_j=a,B_j=b, X_j=x,Y_j=y\mid H_j=h,T_j=t}\\
&\pr{A_j=a,B_j=b \mid X_j=x,Y_j=y,H_j=h,T_j=t} = \sum_{\lambda} q_{\lambda|h,t}\ d_\lambda(a,b|x,y)\ , \\
&\forall \lambda,\ q_{\lambda|h,t} \geq 0\, \\
&\sum_{\lambda} q_{\lambda|h,t} = 1\ ,\\
&\pr{X_j=x\mid H_j=h,T_j=t} = p_x + \tau_A\ ,\\
&\pr{Y_j=y\mid H_j=h,T_j=t} = p_y + \tau_B\ ,
\end{align}
where $\hat{s}_{ab|xyt} \in \{0,1\}$ are the normalized score functions.
Note that we write $q_{\lambda|h,t}$ for a fixed history $H_j = h$ and $T_j = t$, but the LP above does not depend on knowing
$h$ and $t$ as they simply form labels. We remark this is an upper bound, since we allowed for maximum bias.

Linear programs can be solved efficiently, although the number of variables is prohibitively
large. Nevertheless, for games used in experiment the number of inputs and outputs is generally small enough for the LP to be
solved using Mathematica or Matlab.

\subsubsection{An analytical bound: example CHSH}
However, analytical derivation is also viable, and indeed for an existing Bell inequality we can convert the parameters
$\bellmin$ and $\bellmax$ into suitable bounds. For the purpose of illustration we provide a very simple example of this idea
using CHSH, where we derive a bound directly in terms of the probability of winning the game. We remark that is a refinement over the analysis in~\cite{Hensen_15} that
becomes interesting for a larger bias $\tau$, but more cumbersome to read.

Specifically, in Lemma \ref{lem:tool} we derive a tight upper bound on the winning probability of CHSH with imperfect random number generators in an event-ready setup. 
Analogous derivations for other simple inequalities are straightforward. 
For CHSH, the inputs $X_i,Y_i$, outputs $A_i,B_i$ and output of the heralding station $T_i$ take values 0 and 1. If $T_i=0$ the scoring variable $C_i$ takes always the value zero, if $T_i=1$ then $C_i=1$ when $x\cdot y = a\oplus b$ and $C_i=0$ in the remaining cases.

\begin{lemma}
\label{lem:tool}
Let $m \in \Natural$, and let the sequence $(\mathbf A^m,\mathbf B^m,\mathbf X^m,\mathbf Y^m,\mathbf H^m,\mathbf T^m)$ correspond with $m$ \TTs\ of a CHSH heralding experiment. 
Suppose that the null hypothesis holds, i.e., nature is governed by an \LHVM. Given that the predictability of the RNG is $\tau$, we have for all $\ellInd \in \Natural$ with $i\leq m$, any possible history $\HH_\ellInd =\h_\ellInd$ of the experiment, and $T_\ellInd = 1$ that the probability of $C_\ellInd=1$ is upper bounded by
\begin{equation}
\pr{C_\ellInd=1|\HH_\ellInd = \h_\ellInd, T_\ellInd = 1}\leq \bellwin^1\ ,
\end{equation}
where $\bellwin^1=3/4+(\tau-\tau^2)$.
\end{lemma}
\begin{proof}
We first expand the desired term using the definition of $C_\ellInd$ as
\begin{align}
\pr{C_\ellInd=1|\HH_\ellInd = \h_\ellInd,T_\ellInd=1} &= \sum_{\substack{{x,y,z\in\{0,1\}}\\(x,y)\neq(1,1)}}
\pr{A_\ellInd=z,B_\ellInd=z,X_\ellInd=x,Y_\ellInd=y|\HH_\ellInd=\h_\ellInd,T_\ellInd =1}\nonumber\\
&+\sum_{\substack{z\in\{0,1\}}}\pr{A_\ellInd=z,B_\ellInd=z\oplus 1,X_\ellInd=1,Y_\ellInd=1|\HH_\ellInd=\h_\ellInd,T_\ellInd=1}\ .\label{eq:probexpansion1}
\end{align}
We can break these probabilities into simpler terms
\begin{align}
&\p\big(A_\ellInd=a,B_\ellInd=b,X_\ellInd=x,Y_\ellInd=y|\HH_\ellInd=\h_\ellInd, T_\ellInd=1\big)  \nonumber\\
&\qquad =\pr{A_\ellInd=a,X_\ellInd=x|\HH_\ellInd=\h_\ellInd,T_\ellInd=1}\nonumber\\
&\qquad\qquad \cdot\pr{B_\ellInd=b,Y_\ellInd=y|\HH_\ellInd=\h_\ellInd,T_\ellInd=1}\\
&\qquad =\pr{X_\ellInd=x|\HH_\ellInd=\h_\ellInd,T_\ellInd=1}\pr{A_\ellInd=a|X_\ellInd=x,\HH_\ellInd=\h_\ellInd,T_\ellInd=1} \nonumber\\
&\qquad\qquad \cdot \pr{Y_\ellInd=y|\HH_\ellInd=\h_\ellInd,T_\ellInd = 1} \pr{B_\ellInd=b|Y_\ellInd=y,\HH_\ellInd=\h_\ellInd,T_\ellInd=1}\ . 
\end{align}
The first equality followed by the locality condition, the second one simply by the definition of conditional probability.
With this decomposition, we can express \eqref{eq:probexpansion1} as
\begin{align}
\pr{C_\ellInd=1|\HH_\ellInd = \h_\ellInd,T_\ellInd=1} &= \sum_{\substack{{x,y\in\{0,1\}}\\(x,y)\neq(1,1)}}\alpha_{x}\beta_{y}\left(\chi_{x}\gamma_{y}+(1-\chi_{x})(1-\gamma_{y})\right)\nonumber\\
                                                                &\qquad \qquad +\alpha_{1}\beta_{1}\left(\chi_{1}(1-\gamma_{1})+(1-\chi_{1})\gamma_{1}\right)\\
                                                                &= \sum_{x,y\in\{0,1\}} \alpha_{x}\beta_{y}f_{x,y}\ . \label{eq:compactsum}
\end{align}
where we have used the shorthands
\begin{align}
\chi_{x}&:=\pr{A_\ellInd=1|X_\ellInd=x,\HH_\ellInd = \h_\ellInd,T_\ellInd=1}\ , \\
\gamma_{y}&:=\pr{B_\ellInd=1|Y_\ellInd=y,\HH_\ellInd = \h_\ellInd,T_\ellInd=1},\\
\alpha_{x}&:=\pr{X_\ellInd=x|H_\ellInd = h_\ellInd,T_\ellInd=1}\ ,\\
\beta_{y'}&:=\pr{Y_\ellInd=y|H_\ellInd = h_\ellInd,T_\ellInd=1}\ ,\\
f_{x,y} &:= 
\left\{ 
\begin{aligned}
\chi_{x}\gamma_{y}+(1-\chi_{x})(1-\gamma_{y}) & \qquad\textrm{ if } (x,y)\neq (1,1)\ , \\
\chi_{x}(1-\gamma_{y})+(1-\chi_{x})\gamma_{y} & \qquad\textrm{ otherwise.}
\end{aligned}
\right.
\end{align}
Now we will expand \eqref{eq:compactsum}. We know that $1/2 -\tau\leq\alpha_{x},\beta_{y}\leq 1/2 +\tau$. In principle, $\tau$ does not need to take the values in the extreme on the range. Without loss of generality let $\alpha_{0}=1/2+\tau_A$ and $\beta_0=1/2+\tau_B$, with $\tau_A,\tau_B\in(-1/2,1/2)$.
\begin{align}                                                                
\sum_{x,y\in\{0,1\}} \alpha_{x}\beta_{y}f_{x,y}  &= \left(\frac{1}{2}+\tau_{A}\right)\left(\frac{1}{2}+\tau_{B}\right)f_{0,0} + \left(\frac{1}{2}+\tau_{A}\right)\left(\frac{1}{2}-\tau_{B}\right)f_{0,1}\nonumber\\
                                                                &\qquad+\left(\frac{1}{2}-\tau_{A}\right)\left(\frac{1}{2}+\tau_{B}\right)f_{1,0}+\left(\frac{1}{2}-\tau_{A}\right)\left(\frac{1}{2}-\tau_{B}\right)f_{1,1}\\
                                                                &= \left(\frac{1}{4}+\frac{1}{2}\tau_{A}+\frac{1}{2}\tau_{B}+\tau_{A}\tau_{B}\right)f_{0,0} + \left(\frac{1}{4}+\frac{1}{2}\tau_{A}-\frac{1}{2}\tau_{B}-\tau_{A}\tau_{B}\right)f_{0,1}\nonumber\\
                                                                &\qquad+\left(\frac{1}{4}-\frac{1}{2}\tau_{A}+\frac{1}{2}\tau_{B}-\tau_{A}\tau_{B}\right)f_{1,0}+ \left(\frac{1}{4}-\frac{1}{2}\tau_{A}-\frac{1}{2}\tau_{B}+\tau_{A}\tau_{B}\right)f_{1,1}\\
&= \left(\tau_{A}+\tau_{B}\right)f_{0,0} + \left(\tau_{A}-2\tau_{A}\tau_{B}\right)f_{0,1}\nonumber\\
                                                                &\qquad+\left(\tau_{B}-2\tau_{A}\tau_{B}\right)f_{1,0}+ \left(\frac{1}{4}-\frac{1}{2}\tau_{A}-\frac{1}{2}\tau_{B} +\tau_{A}\tau_{B}\right)\sum_{a,b}f_{a,b}\ .\label{eq:FRemain}
\end{align}                                                                
It thus remains to bound the sum of $f_{x,y}$. Note that we can write
\begin{align}
\sum_{x,y\in\{0,1\}} f_{x,y} &= \left(\chi_{0}\gamma_{0}+(1-\chi_{0})(1-\gamma_{0})\right)
                                                               + \left(\chi_{0}\gamma_{1}+(1-\chi_{0})(1-\gamma_{1})\right) \nonumber\\
                                                               &\qquad + \left(\chi_{1}\gamma_{0}+(1-\chi_{1})(1-\gamma_{0})\right)
                                                               + \left(\chi_{1}(1-\gamma_{1})+(1-\chi_{1})\gamma_{1}\right) \\
&= \chi_0 \left(\gamma_0 + \gamma_1\right) + \left(1-\chi_0\right) \left(2 - \gamma_0 - \gamma_1\right)\nonumber\\
&\qquad + \chi_1 \left(\gamma_0 + 1 - \gamma_1\right) + \left(1- \chi_1\right)\left(1 - \gamma_0 + \gamma_1\right)\label{eq:updelta}\ . 
\end{align}
Since \eqref{eq:updelta} is a sum of two convex combinations, it must take its maximum value at one of the extreme points, that is with $\chi_{0}\in\{0,1\}$ and $\chi_{1}\in\{0,1\}$. We can thus consider all four combinations of values for $\chi_0$ and $\chi_1$ given by
\begin{align}
\sum_{x,y\in\{0,1\}} f_{x,y} =
\left\{
\begin{aligned}
3-2\gamma_{0} &\textrm{ if }(\chi_{0},\chi_{1})=(0,0)\ ,\\ 
3-2\gamma_{1} &\textrm{ if }(\chi_{0},\chi_{1})=(0,1)\ ,\\
1+2\gamma_{1} &\textrm{ if }(\chi_{0},\chi_{1})=(1,0)\ ,\\
1+2\gamma_{0} &\textrm{ if }(\chi_{0},\chi_{1})=(1,1)\ .
\end{aligned}
\right.
\end{align}
Since $0 \leq \gamma_0,\gamma_1 \leq 1$, we have in all cases that the sum is upper bounded by $3$. 

Finally, using~\eqref{eq:FRemain} we have
\begin{align}
\pr{C_\ellInd=1|H_\ellInd=h_\ellInd,T_\ellInd=1} &\leq  2\left(\tau_{A}+\tau_{B}-2\tau_{A}\tau_{B}\right)+3\left(\frac{1}{4}-\frac{1}{2}\tau_{A}-\frac{1}{2}\tau_{B} +\tau_{A}\tau_{B}\right)\\
                                                                &=\frac{3}{4}+\frac{1}{2}(\tau_A+\tau_B)-\tau_A\tau_B\label{eq:lotsoftau}\\
                                                                &\leq\frac{3}{4}+\tau-\tau^2 
\end{align}
where in the first inequality we bound $f_{0,0},f_{0,1},f_{1,0}$ by 1. The second inequality follows since $\tau\leq 1/2$ and for $\tau_A,\tau_B$ below $1/2$ \eqref{eq:lotsoftau} is strictly increasing both in $\tau_A$ and $\tau_B$; this implies that the maximum is found in the extreme: $\tau=\tau_A=\tau_B$.
\end{proof}

Note that in the case $T_i=0$, we trivially have $\pr{C_\ellInd=1|\HH_\ellInd = \h_\ellInd,T_\ellInd=0}=0=\bellwin^0$.

\subsubsection{Step 2: Replacing the history with the recorded values of $\mbigCi$ and $\mbigTi$.}
Now, building on the above, we prove that the probability that $C_i$ takes the value one given not the entire history, but only the heralding events and the prior sequence of scores, is bounded from above by the same $\bellwin^t$. 
While the two statements look very similar, the main difference is that while in Step 1 we condition on the entire history $\HH_\ellInd =\h_\ellInd$, in Lemma \ref{lem:tool2} we condition on the heralding events $\mbigTi=\mti$, and the prior sequence $\mbigCi = (C_j)_{j=1}^{i-1}$ of data that can 
actually be observed~\footnote{Note that since the history captures an arbitrary state of the experiment in the past, it
could also include things which are not measured or recorded by the experimenter.}. 
Although both statements are similar, it is Lemma \ref{lem:tool2} that we can easily use in the proof of Lemma \ref{lem:mainth} to bound the $p$-value of the experiment.

We will need Proposition \ref{prop}, which is a basic probabilistic statement necessary for Lemma \ref{lem:tool2}. In essence, it is just the measure theoretic version of
\begin{align}
\pr{A=a}=\sum_{b}\pr{A=a|B=b}\pr{B=b}\ .
\end{align}
We state it for completeness, with the purpose of having the derivation of the bound on the $p$-value as self contained as possible.
\begin{proposition}[Law of total probability]
\label{prop}
Let $A,B$ be two random variables on the same probability space $\Omega$ with $\expt(|A|)<\infty$. Then the probability of an event $A=a$ admits the following integral form
\label{prop:trick}
\begin{align}
\pr{A=a} &= \int_\Omega \pr{A=a|B=b}d\mu(b)\ ,
\end{align}
for some measure $d\mu$ on $\Omega$.
\end{proposition}
\begin{lemma}
\label{lem:tool2}
Suppose that the null hypothesis holds, i.e., nature is governed by an \LHVM. 
Let $m \in \Natural$, and let the sequence $(\mathbf A^m,\mathbf B^m,\mathbf X^m,\mathbf Y^m,\mathbf H^m,\mathbf T^m)$ correspond with $m$ \TTs\ of a heralding experiment.  The heralding station has $n$ outputs. 
If for all $i \in \Natural$ with  $i \leq \N$, any possible history $\HH_\ellInd =\h_\ellInd$ of the experiment, and $T_\ellInd = t_\ellInd$ the probability that $C_i$ takes the value one satisfies: $$\pr{C_\ellInd=1|\HH_\ellInd = \h_\ellInd, T_\ellInd = t_\ellInd}\leq \bellwin^{t_\ellInd}.$$
Then for all sequences $\mti \in \{0,1\}^{\ellInd}$ and $\mvi \in \{0,1\}^{\ellInd-1}$:
\begin{equation}
\pr{C_\ellInd=1|\mbigTi= \mti,\mbigCi = \mvi}\leq \bellwin^{t_\ellInd}\ .
\end{equation}
\end{lemma}
\begin{proof}
The following equalities hold from the definition of conditional probability and Proposition \ref{prop:trick}
\begin{align}
&\pr{C_\ellInd=1|\mbigTi=\mti,\mbigCi = \mvi} \pr{\mbigTi= \mti,\mbigCi =\mvi}\nonumber\\
&\qquad = \pr{C_i=1,\mbigTi= \mti,\mbigCi = \mvi}\\
\label{eq:fromhere}                           &\qquad = \int_\Omega \pr{{C_\ellInd=1,\mbigTi= \mti,\mbigCi = \mvi}|H_\ellInd=\h_\ellInd}d\mu(\h_i)\ .
\end{align}
Let us bound the integrand in the previous equation. We have
\begin{align}
\pr{{C_\ellInd=1,\mbigTi= \mti,\mbigCi = \mvi}|H_\ellInd=\h_\ellInd}
&= \pr{{C_\ellInd=1,T_\ellInd = t_\ellInd}|H_\ellInd=\h_\ellInd} \cdot\deltaDef
\label{eq:EEEfirst}\\
&= \pr{{C_\ellInd=1}|H_\ellInd = \h_\ellInd,T_\ellInd = t_\ellInd} \label{eq:EEEsecond}
\pr{{T_\ellInd = t_\ellInd}|H_\ellInd = \h_\ellInd}\cdot\deltaDef \\
&\leq \bellwin^{t_\ellInd}\, \pr{{T_\ellInd = t_\ellInd}|H_\ellInd = \h_\ellInd}\cdot\deltaDef\label{eq:EEEfourth} \\
\label{eq:plugback}                 &\leq \bellwin^{t_\ellInd}\, \pr{\mbigTi= \mti, \mbigCi = \mvi|H_\ellInd = \h_\ellInd}\ , 
\end{align}
where $\deltaDef$ is a shorthand for
\begin{equation}
\deltaDef = \pr{\mathds 1\{\mbigTiprevious =\mtiprevious,\mbigCi = \mvi\}=1|H_\ellInd=\h_\ellInd}\ .
\end{equation}
The first equality~\eqref{eq:EEEfirst} follows from the fact that $\mtiprevious$ and $\mvi$ are events either compatible or incompatible with $\h_\ellInd$, the second one~\eqref{eq:EEEsecond} 
from the definition of conditional probability, 
and the inequality~\eqref{eq:EEEfourth} from Lemma \ref{lem:tool}. 
We now introduce~\eqref{eq:plugback} back into~\eqref{eq:fromhere} to obtain
\begin{align}
&\pr{C_i=1|\mbigTi=\mti,\mbigCi = \mvi} \pr{\mbigTi= \mti,\mbigCi =\mvi}\nonumber\\
 &\qquad \leq \bellwin^{t_\ellInd}\, \int_\Omega  \pr{{\mbigTi= \mti, \mbigCi = \mvi}|H_\ellInd=\h_\ellInd} d\mu(h_\ellInd)\\ 
                                  &\qquad =  \bellwin^{t_\ellInd}\, \pr{\mbigTi=\mti,\mbigCi = \mvi}\ , \label{eq:FFF}
\end{align}
where the equality~\eqref{eq:FFF} follows from Proposition \ref{prop:trick}. We complete the proof by cancelling the terms\\
 $\pr{\mbigTi=\mti,\mbigCi= \mvi}$ on the right and left side of the equation above.                            
\end{proof}

\subsection{Step 3: Going from one to many \TTs}
In the last part of this technical derivation, we put together the statements above and instead of making a statement just about the next \TT, we now make a statement about all \TTs\ together. This proof generalizes Proposition 4 in \cite{Bierhorst_14} to event-ready schemes. Even though the analysis is more involved, the proof technique follows the same steps as the original one in \cite{Bierhorst_14}. 

What makes the analysis more tricky, is that in an event-ready scheme we have a long sequence of $m$ attempts, and a (potentially much shorter) sequence of $n$ trials, that is, attempts for which $t_\ellInd = 1$. 
It is intuitive that of relevance in the long sequence of $m$ \TTs, is the sequence $\mathbf{C}^{m} = (C_1,\ldots,C_m)$ together with the sequence of event-ready attempts $\mathbf{T}^m = (T_1,\ldots,T_m)$. Recall
that the latter tells us which elements of $\mathbf{C}^m$ are of interest, i.e., can at all be non-zero. To reason about the shorter sequence of $n$ \evs, let us first introduce some notation.
Our goal will be to define a series of random variables $\mathbf{D}^n = (D_1,\ldots,D_n)$ for the short sequence of \evs, where intuitively $D_j$ corresponds to the random variable
taking value one when the $j$-th event-ready success also results in $T_i = 1$ for any corresponding $i$.
In other words, we will define $\mathbf{D}^n$ in such a way that instead of worrying
about the number of 1's in $(C_1 T_1,\ldots,C_m T_m)$ we will be concerned with the number of 1's in $(D_1,\ldots,D_n)$.

To define this formally, we need a way to map the $j$-th trial from the short sequence of \evs, to the
index $i$ in the longer sequence of \TTs. Note that for a particular event-ready sequence $\mtotal = (t_1,\ldots,t_m) \in \mbigTtotal$, the $j$-th trial is mapped to the smallest index $i$, such that
the subsequence $\mti = (t_1,\ldots,t_i)$ of $\mtotal$ has exactly $j$ 1's. Of course, there are many sequences $\mti \in \mbigTi$ that have precisely $j$ 1's, where the last element is also a 1, and for all such strings the mapping from $j$ in the sequence of \evs, to the index $i$ in the sequence of \TTs\ is the same.
Let us thus define
\begin{align}
\longS = \left\{\mtotal = (t_1,\ldots,t_m) \in \{0,1\}^m\mid |\mtotal| = n \mbox{ and } \mti = (t_1,\ldots,t_i) \mbox{ satisfies } |\mti| = j \mbox{ and } |\mtiprevious| = j - 1\right\}\ ,
\end{align}
to be the set of all event-ready sequences $\mtotal$ for which $j$ is mapped to one particular $i$.
By summing over all possible indices $i$ in the long sequence of \TTs, we can thus formally define
\begin{align}
D_j = \sum_{i=1}^m\ \sum_{\mtotal \in \longS} \sum_{\mathbf{c}^m \in \{0,1\}^m} \mathds{1}\left\{\mbigTtotal = \mtotal, \mathbf{C}^m = \mathbf{c}^m\right\} \cdot C_i\ ,
\end{align}
where $\mathds{1}$ is as before the indicator function.
In terms of probabilities, this means that the probability that the $j$-th trial gives $D_j = 1$ is given by
\begin{align}\label{eq:Dj1}
\pr{D_j = 1} = \sum_{i=1}^m \left(\sum_{\mtotal \in \longS}\ \sum_{\substack{\mathbf{c}^m \in \{0,1\}^m\\c_i = 1}} \pr{\mbigTtotal = \mtotal, \mathbf{C}^m = \mathbf{c}^m}\right)\ .
\end{align}
We can thus express the \pvalues as
\begin{align}
P\textrm{-value} &= \sum_{\substack{\mtotal \in \{0,1\}^m\\|\mtotal| = n}} \pr{\mbigTtotal = \mtotal} \pr{\mbox{number of 1's in}\ (C_1\cdot t_1,\ldots, C_m\ldots t_m) \geq k \mid \mbigTtotal = \mtotal}\\
&= \pr{\mbox{number of 1's in } (D_1,\ldots, D_n) \geq k}\\
&= \pr{\sum_{j=1}^nD_j\geq k}\ .
\end{align}

Before delving into the proof below, it will be convenient to simplify~\eqref{eq:Dj1}. Note that for a fixed $i$, the term in brackets in~\eqref{eq:Dj1} contains a sum over all possible $t_{i+1},\ldots,t_m$ and $c_{i+1},\ldots,c_{m}$.
This means we can use the law of total probability to shorten the sum by expressing~\eqref{eq:Dj1} in terms of the marginal distributions as
\begin{align}\label{eq:Dj1Marginal}
\pr{D_j = 1} = \sum_{i=1}^m\ \sum_{\mti \in \shortS}\ \sum_{\substack{\mvi \in \{0,1\}^{i-1}}} \pr{\mbigTi= \mti, \mbigCi = \mvi, C_i = 1}\ ,
\end{align}
where
\begin{align}
\shortS = \left\{\mti = (t_1,\ldots,t_i) \in \{0,1\}^i \mid \exists \mtotalHat = (\hat{t}_1,\ldots,\hat{t}_m) \in \longS \mbox{ such that } (\hat{t}_1,\ldots,\hat{t}_i) = (t_1,\ldots,t_i)\right\}\ .
\end{align}
After having formally established the relation between the sequence of \evs\ and the sequence of \TTs, we are now ready for the final proof, where
we can now argue in terms of the sequence of \evs\ $(D_1,\ldots,D_n)$.

\begin{lemma}
\label{lem:mainth}
Suppose that the null hypothesis holds, i.e., nature is governed by an \LHVM. 
Let $m,n,k \in \Natural$ and let the sequence $(\mathbf A^m,\mathbf B^m,\mathbf X^m,\mathbf Y^m,\mathbf H^m,\mathbf T^m)$ correspond with $m$ \TTs\ of an event-ready experiment. If
\begin{equation}
\pr{C_\ellInd=1|\mbigTi=\mti,\mbigCi = \mvi}\leq \bellwin,
\end{equation}
then we have that for all $\N \geq n$,
the probability that at least $k$ of the $(D_j)_{j=1}^n$ take the value one is upper bounded by
\begin{equation}
\label{eq:mainthb}
P\textrm{\textnormal{-value}} = \pr{\sum_{j=1}^n D_j\geq k}
\leq P_{n,k}(\B)\ ,
\end{equation}
where $P_{n,k}(\B)$ denotes the probability that $n$ Bernoullis with probability $\gamma := \bellwin$ yield at least $k$ 1's,
and $P_{n,k}(\B)=0$ if $k>n$.
\end{lemma}
\begin{proof}
Let us define the shorthand
\begin{align}
P_{n,k}(D) = \pr{\sum_{j=1}^n D_j \geq k}\ .
\end{align}
The probability that we see at least zero 1's ($k=0$) obeys
\begin{align}
P_{n,0}\left(D\right) &=1 \\
                                      &= P_{n,0}(\B)
\end{align}
for all $n$ and $\N \geq n$.

We now prove the statement for $k>0$ by induction on $n$.
For $n=1$, we need only to verify that~\eqref{eq:mainthb} holds for $k=1$ (we already dealt with $k=0$ and the case $k> 1$ trivially holds). We have
\begin{align}
P_{1,1}\left(D\right) &=\pr{D_1\geq 1} =\pr{D_1= 1}\\
                                     &=\sum_{i=1}^m\ \sum_{\mathbf t^i\in\shortSindex{1}}
                                          \sum_{\mathbf c^{i-1}\in \{0,1\}^{i-1}}
                                          \pr{\mbigTi = \mti, \mathbf C^{i-1}=\mathbf c^{i-1}, C_i=1} \label{eq:OO1} \\
                                     & \label{eq:OO2} =\sum_{i=1}^m\ \sum_{\mathbf t^i\in\shortSindex{1}}
                                          \sum_{\mathbf c^{i-1}\in \{0,1\}^{i-1}}
\pr{C_i=1|\mathbf C^{i-1}=\mathbf c^{i-1}, \mathbf T^i = \mathbf t^i}\pr{\mathbf C^{i-1}=\mathbf c^{i-1}, \mathbf T^i = \mathbf t^i}\\
                                     &\label{eq:OO3} \leq \bellwin\ \sum_{i=1}^m\ \sum_{\mathbf t^i\in\shortSindex{1}}\sum_{\mathbf c^{i-1} \in \{0,1\}^{i-1}}
                                     \pr{\mathbf C^{i-1}=\mathbf c^{i-1}, \mathbf T^j = \mathbf t^j}\\
                                     &=\bellwin = P_{1,1}(\B)\ , \label{eq:OO4}
\end{align}
where the first equality~\eqref{eq:OO1} is just~\eqref{eq:Dj1Marginal}, the second equality~\eqref{eq:OO2} the definition of conditional probability,
the inequality~\eqref{eq:OO3} follows from Lemma~\ref{lem:tool2}, and the final equality~\eqref{eq:OO4} from the definition of the sets $\longS$ and the fact the sum of all probabilities is 1.

In order to prove the induction step below, let us first
express the probability of having at least $k$ 1's on \ev\ $n$ as the sum of the probability of having at least $k$ on \ev\ $n-1$,
plus the probability of having exactly $k-1$ 1's on \ev\ $n-1$ and a one on the $n$-th \ev
\begin{align}
P_{n,k}\left(D\right) &= \pr{\sum_{j=1}^nD_j\geq k} \\
                                     &= \pr{\sum_{j=1}^{n-1}D_j\geq k} + \pr{\sum_{j=1}^{n-1}D_j= k-1,D_n=1} \\
\label{eq:psumdib}          &= P_{n-1,k}\left(D\right) + \pr{\sum_{j=1}^{n-1}D_j= k-1,D_n=1}\ .
\end{align}
We now upper bound the second term in \eqref{eq:psumdib}, where we will use the shorthand $|\mathbf{c}^{i-1} \cdot \mathbf{t}^{i-1}| = |(c_1t_1,\ldots,c_{i-1} t_{i-1})| = \sum_{j=1}^{i-1} c_j t_j$.
\begin{align}
& \pr{\sum_{j=1}^{n-1}D_j= k-1,D_n=1}\\
&\qquad=\sum_{i={1}}^m\ \sum_{\mathbf t^i\in\shortSindex{n}}\ \sum_{\substack{\mathbf{c}^{i-1}\in \{0,1\}^{i-1}\\|\mathbf{c}^{i-1}\cdot\mathbf{t}^{i-1}|=k-1}} \pr{\mathbf{T}^i=\mathbf t^i,\mathbf{C}^{i-1} = \mathbf{c}^{i-1},C_{i}=1} \\
&\qquad=\sum_{i={1}}^m\ \sum_{\mathbf t^i\in\shortSindex{n}}\ \sum_{\substack{\mathbf{c}^{i-1}\in \{0,1\}^{i-1}\\|\mathbf{c}^{i-1}\cdot\mathbf{t}^{i-1}|=k-1}} \pr{C_{i}=1|\mathbf{T}^i=\mathbf t^i,\mathbf{C}^{i-1} = \mathbf{c}^{i-1}}\pr{\mathbf{T}^i=\mathbf t^i,\mathbf{C}^{i-1} = \mathbf{c}^{i-1}}  \label{eq:GGGfirst}\\
\label{eq:mainb}&\qquad\leq\bellwin\ \sum_{i={1}}^m\ \sum_{\mathbf t^i\in\shortSindex{n}}\ \sum_{\substack{\mathbf{c}^{i-1}\in \{0,1\}^{i-1}\\|\mathbf{c}^{i-1}\cdot\mathbf{t}^{i-1}|=k-1}} \pr{\mathbf{T}^i=\mathbf t^i,\mathbf{C}^{i-1} = \mathbf{c}^{i-1}}\\
\label{eq:nextStep}&\qquad= \bellwin\ \pr{\sum_{j=1}^{n-1} D_j = k-1}\\
\label{eq:pnkineq}&\qquad =\bellwin\  \left(P_{n-1,k-1}(D)-P_{n-1,k}(D)\right)\ .
\end{align}
Equality~\eqref{eq:GGGfirst} follows by the definition of conditional probability,
inequality~\eqref{eq:mainb} from Lemma \ref{lem:tool2}, equality~\eqref{eq:nextStep} from~\eqref{eq:Dj1},
and the last equality~\eqref{eq:pnkineq} because the sum over all vectors having exactly $k-1$ $1$'s
equals the probability of having at least $k-1$ 1's, minus the probability of having at least $k$.

Recall that $P_{n,k}(\B)$ stands for the probability of having at least $k$ successes over $n$ Bernoullis with probability $\bellwin$.
Before proving the induction step, we need to rewrite $P_{n,k}(\B)$ as follows
\begin{align}
P_{n,k}(\B) &= \pr{\textrm{at least $k$ successes over $n-1$ Bernoullis}}\nonumber\\
&\quad+\pr{\textrm{exactly $k-1$ successes over $n-1$ Bernoullis and success on the $n$-th }} \\
&=P_{n-1,k}(\B) \nonumber\\
&\quad+ \pr{\textrm{exactly $k-1$ successes over $n-1$ Bernoullis and success on the $n$-th }} \\
&= P_{n-1,k}(\B) \nonumber\\
&\quad+\pr{\textrm{exactly $k-1$ successes over $n-1$ Bernoullis}}\pr{\textrm{success on the $n$-th \ev}} \\
&= P_{n-1,k}(\B) +\left(P_{n-1,k-1}(\B)-P_{n-1,k}(\B)\right)\pr{\textrm{success on the $n$-th \ev}}\\
\label{eq:laststep}&=P_{n-1,k}(\B) +\bellwin\left(P_{n-1,k-1}(\B)-P_{n-1,k}(\B)\right)\ .
\end{align}
Now we prove the induction step. Consider some arbitrary $n>1$ and consider the induction hypothesis that $\forall m \geq n$ and $\forall k' \leq n$ ,  $P_{n-1,k'}\left(D\right)\leq  P_{n-1,k'}(\B)$. The following chain of inequalities hold:
\begin{align}
P_{n,k}\left(D\right) &\leq P_{n-1,k}\left(D\right) + \bellwin\  \left(P_{n-1,k-1}(D)-P_{n-1,k}(D)\right) \label{eq:F1}\\
                                    &= (1-\bellwin)\ P_{n-1,k}\left(D\right) + \bellwin\ P_{n-1,k-1}(D) \label{eq:F2} \\
                                    &\leq (1-\bellwin)\ P_{n-1,k}\left(\B\right) + \bellwin\ P_{n-1,k-1}(\B)\label{eq:F3}\\
                                    &=P_{n,k}(\B)\ . \label{eq:F4}
\end{align}
The first~\eqref{eq:F1} and second~\eqref{eq:F2} equalities follow after plugging~\eqref{eq:pnkineq} back into \eqref{eq:psumdib} and rearranging. The inequality~\eqref{eq:F3}
follows from the induction hypothesis. The last equality~\eqref{eq:F4} follows from rearranging \eqref{eq:laststep} and completes the induction step.
\end{proof}

\subsubsection{Event-ready schemes creating multiple states}

We now consider arbitrary event-ready schemes in which multiple games can be played as specified by the tag $t$. This is of interest, for example, when multiple Bell states
can be generated at the event-ready station, but we want to use more than one of them. Note that this only makes sense if Alice and Bob can hope to violate all inequalities
using the same measurements, since the event-ready signal is space-like separated from the generation of the random numbers used as inputs to Alice and Bob.
In the Delft experiment~\cite{Hensen_15}, only one Bell state was used, namely the one which is most noise-free resulting in a higher violation.

However, it is possible to use multiple Bell states by playing the standard CHSH game $t=1$ for one Bell state, and one in which we flip the role of the inputs for $t=2$. That is,
instead of taking $x \cdot y = a \oplus b$ as the winning condition, we take $x \cdot y = a \oplus b \oplus 1$. This game has exactly the same winning probability than the
standard CHSH game, meaning that our analysis above applies without change by using a new tag $t'=0$ whenever $t=0$, and setting $t'=1$ for either $t=1$ or $t=2$. Naturally, we then need to apply the relabelling to the outputs before computing the total score.

We remark that is possible to combine games that have different success probabilities, but we are not aware of any situation yet in which this may be beneficial.

\subsection{Conventional analysis}
For completeness, let us now illustrate how the $p$-value for win/lose games compares to statements made in a conventional analysis using standard deviations. In such an analysis it is assumed that there is no memory and that the distribution is Gaussian.
For simplicity, we thereby consider the case where no event-ready scheme is used. It is straightforward to extend to event-ready schemes.
Observe that in win/lose games without the use of an event-ready procedure we can express the $P$-value as
\begin{align}
P\textrm{-value} 
&= \pr{C \geq c}\ , \label{eq:rhs}
\end{align}
where $n$ is the number of trials.
Now, assume that each of the $C_\ellInd$ is i.i.d. (independently and identically distributed) 
and characterized by the probability that it takes the value one
\begin{equation}
\forall i,\ \pr{C_\ellInd=1}=\myprob\ .
\end{equation}
We can now approximate the sum of Bernoulli trials by a Gaussian random variable with mean $n\myprob$ and variance $n\myprob(1-\myprob)$. If the hypothesis holds we can approximate the right hand side of \eqref{eq:rhs}. However, for all win/lose games we have that $\myprob\leq\bellwin$, that is $\bellwin$ is a cap on the probability that $C_\ellInd$ takes the value one. If additionally $c>n\gamma$ we obtain an upper bound on the approximation as follows
\begin{align}
\pr{\sum_{\ellInd=1}^nC_\ellInd\geq c} &\approx Q\left(\frac{c-n\myprob}{\sqrt{n\myprob(1-\myprob)}}\right)\label{eq:firstq}\\
                                            &\leq Q\left(\frac{c-n\bellwin}{\sqrt{n\bellwin(1-\bellwin)}}\right)
\end{align}
where $Q(\cdot)$ denotes the tail probability of the standard normal distribution. 
Observe that the right hand side of \eqref{eq:firstq} is increasing in $\myprob$ (or alternatively $\frac{c-n\myprob}{\sqrt{n\myprob(1-\myprob)}}$ is decreasing in $\myprob$). That is, the inequality follows because $\bellwin$ is the largest possible value of $\myprob$.

\subsection{Relation to the CHSH correlator $\langle CHSH \rangle$}\label{sec:usualCHSH}
For completeness, let us explain how for the example of the CHSH correlator, one can understand the relation between the 
number of wins $c$ obtained in the win/lose game, and the maybe more familiar form of the correlator.
Since our objective is only to illustrate this link and give some intuition on the $p$-values, we assume, only from here and until the end of this section, perfect RNGs. We
will also drop the index $\ellInd$, since we are considering only one trial.
We denote by $\langle XY\rangle_{ab}$ the average of the random variable $XY$ when the settings are $A=a,B=b$
\begin{align}
\langle XY\rangle_{ab} &= \pr{X=Y|A=a,B=b}-\pr{X\neq Y|A=a,B=b}\\
                                     &= \left\{ 
\begin{aligned}
2\pr{C=1|A=a,B=b} - 1 & \qquad\textrm{ if } (a,b)\neq (1,1)\ , \\
1-2\pr{C=1|A=a,B=b}  & \qquad\textrm{ otherwise.}
\end{aligned}
\right.
\end{align}
Let us denote by $S$ the average CHSH value
\begin{align}
S &= \langle CHSH \rangle = \langle XY\rangle_{00}+\langle XY\rangle_{01}+\langle XY\rangle_{10}-\langle XY\rangle_{11}\ .
\end{align}
Now, we can link $S$ with $\pr{C=1}$ as
\begin{align}
\frac{S+4}{8} &= \frac{2\sum_{a,b}\pr{C=1|A=a,B=b}}{8}\\
  &= \sum_{a,b}\frac{1}{4}\pr{C=1|A=a,B=b}\\  
  &= \pr{C=1}\ .
\end{align}
That is, we can map the average CHSH value $S$ to the probability that $C$ takes the value one. It directly follows that the known CHSH upper bound $S\leq 2$ corresponds with $\pr{C=1}\leq 0.75$, which is the upper bound that we obtain if we assume perfect RNGs ($\tau=0$) in Lemma \ref{lem:tool} and Lemma \ref{lem:tool2}. In the same way we can map the observed CHSH violation to the number of successes. Let $n_{a,b}$, $n_{a,b}^{=}$ and $n_{a,b}^{\neq}$ denote 
denote the number of \evs,
the number of wins
and the number of losses associated with setting $(a,b)$ and let $\tilde S$ denote the observed CHSH value: 
\begin{align}
\tilde S &= \sum_{a,b} \frac{n_{a,b}^{=}-n_{a,b}^{\neq}}{n_{a,b}}\ .
\end{align}
For large $n$, the following equalities hold approximately
\begin{align}
n\cdot\frac{\tilde S +4}{8} &= n\cdot\frac{\sum_{a,b} \frac{n_{a,b}^{=}-n_{a,b}^{\neq}}{n_{a,b}}+4}{8}\\
                                      &= n\cdot\frac{\sum_{a,b} \left(2\frac{n_{a,b}^{=}}{n_{a,b}}-1\right)+4}{8}\\
                                      &\approx \sum_{a,b} n_{a,b}^{=}\\
                                      &= c\ . 
\end{align}
The approximation holds since for large $n$ the number of \evs\ at each setting should be approximately $n/4$ and in consequence $n/n_{a,b}\approx 4$. 

\section{Analysis of general games}\label{sec:General}
In this section we deal with games where the score functions $s_{ab|xy}$ take more than two values. We will drop the index $t$ and consider only one game, but it is straightforward to
adapt the following analysis to event-ready schemes as above.
We will use the shorthands
\begin{align}
s_{\max}&=\max_{a,b,x,y}s_{ab|xy}\\
s_{\min}&=\min_{a,b,x,y}s_{ab|xy}
\end{align}
to denote the maximum and minimum values of the score functions.
Recall that since
\begin{align}
s_{ab|xy} &= \frac{s^{xy}_{ab}}{p(x,y)}\ , 
\end{align}
the values of $s_{\max}$ and $s_{\min}$ depend on the distribution $p(x,y)$, if we use imperfect RNGs then bounds on the bias of the random numbers will translate into 
different values of $s_{\max}$ and $s_{\min}$.

In order to obtain a bound on the \pvalues we will proceed as follows. First we state Bentkus' inequality, a concentration bound for bounded martingale difference sequences. Then we show how to construct a martingale sequence with bounded differences for any Bell inequality and finally apply Bentkus' inequality to this martingale sequence.

\subsection{Bentkus' inequality}
The score functions of general games take more than two values, and in full generality they might even take a continuous range of values. However, Bentkus' is given in terms of the tail of a sum of independent and identically distributed Bernoulli random variables or, as we equivalently state it here, in terms of the tail of a binomial distribution. These random variables are discrete and, by definition, can only take a discrete number of values. The gap between both is bridged by an interpolation of the binomial distribution. 
Let us introduce $\mathring{P}_{n,y}\left(\B\right)$ a function that interpolates $P_{n,k}(\B)$ between consecutive values of $k$. We define $\mathring{P}_{n,y}\left(\B\right)$:
\begin{equation}
\mathring{P}_{n,y}\left(\B\right) = \left(P_{n,\lfloor y\rfloor}(\B)\right)^{1- (y - \lfloor y\rfloor)}\left(P_{n,\lceil y\rceil}(\B)\right)^{y - \lfloor y\rfloor}
\end{equation}

\begin{theorem}[Bentkus' inequality (Theorem 1.2 \cite{Bentkus_04})]
Let $M_n$ be a martingale sequence with differences $X_k=M_k-M_{k-1}$ and $M_0=0$. If for $k=1...n$ the differences satisfy the following boundedness condition:
\begin{equation}
\pr{-\alpha_k\leq X_k\leq 1-\alpha_k}=1
\end{equation}
then
\begin{equation}
\pr{M_n\geq x} \leq e \mathring{P}_{n,x+n \gamma}\left(\mathbb B_\gamma \right)
\end{equation}
and $\gamma =\sum_{\ellInd=1}^n \alpha_\ellInd/n$.
\end{theorem}
That is, Bentkus' inequality states that the probability that a martingale sequence surpasses some value $x$ for any martingale sequence with bounded differences is bounded by the tail of a binomial distribution multiplied by $e$. The inequality also holds if $M_n$ is a sequence of supermartingale differences \cite{Bentkus_06}. 

\subsection{A bounded difference supermartingale sequence}
Let $\{C_\ellInd\}_{\ellInd=1}^n$ denote the scores at each \TT\ of a game (see Section \ref{sec:notation}) and let $\mean\left[\bellF_\ellInd|H_\ellInd=\h_\ellInd\right]$ stand for the mean value of the score function given the history. This mean is nothing more than a Bell inequality and we have that 
\begin{align}
\bellmin \leq \mean\left[\bellF_\ellInd|H_\ellInd=\h_\ellInd\right]\leq\bellmax\ .
\end{align}
Consider the sequence of random variables
\begin{equation}
M_n=\sum_\ellInd^{n}\frac{\bellF_\ellInd-\bellmax}{s_{\max}-s_{\min}}\ ,
\end{equation}
the elements in the sequence correspond to the scores in each \TT\ normalized by $s_{\max}-s_{\min}$ and displaced by 
$\bellmax$. 
We subtract 
$\bellmax$ 
such that $M_n$ is a supermartingale sequence with respect to the sequence $H_n$. This can be readily verified since:
\begin{equation}
\mean[M_n-M_{n-1}|H_n\ldots H_1]=\frac{1}{s_{\max}-s_{\min}}\mean\left[\bellF_n-\bellmax|H_n\ldots H_1\right]\leq 0
\end{equation}

Now let $\{c_\ellInd\}_{\ellInd=1}^n$ denote a sequence of observations which correspond to instances of $\{\bellF_\ellInd\}_{\ellInd=1}^n$. We can evaluate the probability that the sequence of random variables $\{\bellF_\ellInd\}_{\ellInd=1}^n$ would yield values equal or higher than $\{c_\ellInd\}_{\ellInd=1}^n$:
\begin{align}
\pr{\sum_\ellInd^{n}\bellF_\ellInd\geq\sum_\ellInd^{n}c_\ellInd} &= \pr{\sum_\ellInd^{n}\bellF_\ellInd -\bellmax\geq\sum_\ellInd^{n}c_\ellInd-\bellmax}\\   
    &= \pr{\frac{\sum_\ellInd^{n}\bellF_\ellInd -\bellmax}{s_{\max}-s_{\min}}\geq\frac{\sum_\ellInd^{n}c_\ellInd-\bellmax}{s_{\max}-s_{\min}}}\\
    &= \pr{M_n\geq\frac{\sum_i^{n}c_\ellInd-\bellmax}{s_{\max}-s_{\min}}}\\
    &=\pr{M_n\geq x}\ ,
\end{align}
where in the last equation we introduced the shorthand 
\begin{equation}
\label{eq:id}
x := \sum_\ellInd^{n}\frac{c_\ellInd-\bellmax}{s_{\max}-s_{\min}}.
\end{equation}
The martingale difference $M_k-M_{k-1}$ is bounded as follows
\begin{equation}
\frac{s_{\min}-\bellmax}{s_{\max}-s_{\min}}\leq M_k-M_{k-1}\leq \frac{s_{\max}-\bellmax}{s_{\max}-s_{\min}}
\end{equation}
Let us denote the lower bound by
\begin{equation}
-\alpha_k=\frac{s_{\min}-\bellmax}{s_{\max}-s_{\min}}\ ,
\end{equation}
and the upper bound by
\begin{equation}
1-\alpha_k=\frac{s_{\max}-\bellmax}{s_{\max}-s_{\min}}\ .
\end{equation}

\subsubsection{Application of Bentkus' inequality to Bell experiments}
With the supermartingale differences bounded above and below by $1-\alpha_k$ and $-\alpha_k$ respectively we can apply Bentkus' inequality:
\begin{align}
\pr{\sum_\ellInd^{n}\bellF_\ellInd\geq\sum_\ellInd^{n}c_\ellInd} &\leq e \mathring{P}_{n,x+n\gamma}(\mathbb B_\gamma)
\label{eq:bentkusmar} 
\end{align}
Where we use the identification of $x$ from \eqref{eq:id}. 
In order to evaluate the right hand side of \eqref{eq:bentkusmar}, there remains to identify 
$x+n\gamma$. 
It can be evaluated from the observed data:
\begin{align}
x+n\gamma &= x + \sum_{\ellInd=1}^n\frac{\bellmax-s_{\min}}{s_{\max}-s_{\min}}\\
         &= \sum_{\ellInd=1}^{n}\frac{c_\ellInd-\bellmax}{s_{\max}-s_{\min}} + \sum_{\ellInd=1}^n\frac{\bellmax-s_{\min}}{s_{\max}-s_{\min}}\\
         &= \sum_{\ellInd=1}^{n}\frac{c_\ellInd-s_{\min}}{s_{\max}-s_{\min}} \ .
\end{align}
Finally, 
we obtain the following bound by directly applyng Bentkus' inequality 
\begin{align}
\pr{C \geq c} &\leq e\ \mathring{P}_{n,x+n\gamma}(\mathbb B_{\hat \gamma})
\label{eq:bentkusmarfinal} 
\end{align}

\end{document}